\documentclass[letterpaper]{article}
\usepackage[top=1in, bottom=0.9in, left=0.9in, right=1in]{geometry}
\usepackage{parskip} 
\usepackage{amssymb}
\usepackage{makecell}

\usepackage{fullpage}
\usepackage[dvips]{color}
\usepackage{xcolor,float}
\usepackage{latexsym}
\usepackage[activeacute,english]{babel}
\usepackage{graphicx}
\usepackage{dsfont}
\usepackage{amsmath}
\usepackage{bm,bbm}
\usepackage{nccmath}
\usepackage{natbib}
\usepackage{amsthm,appendix,url,inconsolata,wrapfig}

\usepackage{algorithm}
\usepackage{algpseudocode}

\usepackage{mathtools}
\usepackage{enumitem}

\usepackage{multirow,sidecap,wrapfig}

\usepackage{amscd}
\usepackage{amsfonts}
\usepackage{amssymb}
\usepackage[tableposition=top]{caption}
\usepackage{wrapfig}
\usepackage{ifthen}
\usepackage[utf8]{inputenc}
\usepackage{ulem} 
\usepackage{xspace} 

\usepackage{afterpage}
\usepackage{setspace}
\usepackage[customcolors]{hf-tikz}
\usepackage{authblk} 
\usepackage[hidelinks]{hyperref}

\usepackage{caption}
\usepackage{subcaption}

\newtheoremstyle{mytheoremstyle}
  {\topsep} 
  {\topsep} 
  {} 
  {} 
  {\bfseries} 
  {.} 
  {.5em} 
  {} 
\theoremstyle{mytheoremstyle}
\newtheorem{theorem}{Theorem}[section]
\newtheorem{definition}[theorem]{Definition}

\allowdisplaybreaks[4]

\DeclarePairedDelimiter{\abs}{\lvert}{\rvert}

\usepackage{booktabs}
\usepackage{colortbl}

\colorlet{tableheadcolor}{gray!25} 

\newcolumntype{K}[1]{>{\centering\arraybackslash}m{#1}}

\title{Matrix representations and distance metrics for unlabeled ranked phylogenetic networks}
\author[1,2]{Jiayang Wang}
\author[3,4,*]{Julia A. Palacios}
\author[2,5,*]{Claudia Sol\'is-Lemus}
\affil[1]{Department of Statistics, University of Wisconsin-Madison, Madison, WI 53706, USA}
\affil[2]{Wisconsin Institute for Discovery, University of Wisconsin-Madison, Madison, WI 53706, USA}
\affil[3]{Department of Statistics, Stanford University, Stanford, CA 94305, USA}
\affil[4]{Department of Biomedical Data Science, Stanford Medicine, Stanford, CA 94305, USA}
\affil[5]{Department of Plant Pathology, University of Wisconsin-Madison, Madison, WI 53706, USA}
\affil[*]{Corresponding authors. Email: juliapr@stanford.edu, solislemus@wisc.edu}
\date{}

\graphicspath{ {figure/} }

\begin{document}
\onehalfspacing
\maketitle


\begin{abstract}
Phylogenetic networks are graphs inferred from molecular sequence data that represent ancestral histories shaped by reticulate processes such as recombination, hybridization, and horizontal gene transfer. 
We introduce a family of distance metrics for rooted, ranked, unlabeled phylogenetic networks, extending a previously developed distance for ranked trees. Our approach relies on a bijective triangular matrix representation of phylogenetic networks that captures the temporal order of internal events, speciations, and hybridizations.  Our metrics, defined as standard matrix norms, allow efficient quantitative comparisons of network topologies, timed networks and networks with differing numbers of hybridizations. Our distance can be used for both isochronous networks where all tips are sampled at one time point, and heterochronous networks where tips are allowed to be sampled at different time points. We show that our metrics capture biologically meaningful differences among evolutionary histories in both simulations and empirical posterior distributions of viral phylogenetic networks.  
These tools fill a methodological gap, enabling principled comparisons of ranked, unlabeled phylogenetic networks, including ancestral recombination graphs.


\end{abstract}


\section{Introduction}


The ancestral history of a sample of molecular sequences at a particular genetic segment is commonly represented by a binary phylogenetic tree; however, trees are insufficient when genomic regions experience reticulation events such as hybridization, horizontal gene transfer, recombination, or reassortment \citep{kong2025phylogenetic}. Phylogenetic networks extend trees by allowing reticulation events and thus provide a more realistic representation of ancestry for many organisms such as hybrid speciation in plants \citep{morales-briones2018, marcussen2014}, introgression in animals \citep{Yu2014, SolisLemus2016}, bacterial gene flow \citep{abby2012}, and reassortment or recombination in viruses \citep{Muller2020, forster2020, Sanchez-Pacheco2020, Mavian2020}. 

In this work, we focus on ranked unlabeled phylogenetic networks which are rooted acyclic graphs whose internal nodes are ordered (ranked) in time, and whose leaves (sampled taxa or sequences) are unlabeled. We consider both, the network topology only, as well as the timed phylogenetic network with branch lengths. 
%
Ranked unlabeled networks are critical for evolutionary applications. 
Ranking (the temporal ordering of internal events) is intrinsic to many inference frameworks such as the coalescent and birth–death models and encodes evolutionary timing that is lost in unranked graph shapes. Notably, ancestral recombination graphs (ARGs) are phylogenetic networks that represent ancestral histories with recombination  and play a central role in evolutionary inference in population genetics \citep{griffiths1997ancestral}. Similarly, reassortment graphs are phylogenetic networks that represent the evolutionary history with reassortment of viral segments and play an important role in understanding disease dynamics \citep{Muller2020}. 

Motivated by the recent development of distance metrics on ranked, unlabeled phylogenetic trees via matrix encodings \citep{Kim2020}, we extend the metric notion to rooted, ranked, unlabeled phylogenetic networks. 
Although numerous metrics have been proposed for labeled phylogenetic trees \citep{KuhnerYam} and, to a lesser extent, for certain classes of labeled networks \citep{cardona2008metrics, cardona2008distance, cardona2008metrics2, nakhleh2009metric, Huson2010, moulton2017cubic, yakici2022phylogenetic, maxfield2025dissimilarity}, there are currently no metrics designed specifically for ranked, unlabeled phylogenetic networks. Existing approaches typically require shared taxon labels, thus preventing comparisons across non-overlapping samples. Alternatively, the metrics tend to ignore ranking and therefore discard essential temporal information, or they become computationally infeasible for large posterior samples. This gap restricts our ability to quantitatively compare inferred network topologies (and timed networks) across datasets, summarize posterior distributions of networks, or rigorously evaluate estimation procedures. 

A useful metric on the space of ranked unlabeled networks should therefore (i) encode both topology and the temporal ordering of internal events, (ii) accommodate hybridizations (possibly differing in number between networks), (iii) support comparisons of heterochronous samples (different sampling times) as well as isochronous ones, and (iv) be computationally efficient enough to apply to large collections of inferred networks. In particular, analyses of rapidly evolving pathogens, including viruses, often rely on heterochronous sampling, with samples obtained across different time points.  With these considerations in mind, we develop a family of distance metrics that relies on extending the $\mathbf F$-matrix encoding used for ranked trees \citep{Kim2020} to a new matrix representation suitable for networks.
Finally, we define the distance between two networks as a matrix norm of the difference between the two corresponding  $\mathbf F$-matrices.

Concretely, our contributions are as follows.
We introduce a novel triangular matrix representation called $\mathbf{F}$-matrix for rooted, ranked, unlabeled phylogenetic networks that captures internal node ordering, branch length information, and the presence and placement of hybridizations. In the absence of hybridization events, the $\mathbf{F}$-matrix representation of the unlabeled phylogenetic tree reduces to the $\mathbf{F}$-matrix definition in \citet{Kim2020}. An important property of this representation is that the space of $\mathbf F$-matrices is defined as the space of triangular matrices subject to integer and linear constraints. This allows us to easily enumerate the whole space of ranked unlabeled phylogenetic networks, study combinatorial aspects of the phylogenetic space and define summary statistics, such as Frech\'{e}t mean, that can be obtained by solving linear integer optimization problems.  

We use this representation to define a family of distances on ranked networks by applying standard distances defined on matrices. These metrics inherit favorable computational properties from the matrix representation and can be evaluated in time quadratic in the number of leaves in the typical case.
We extend the framework to handle networks with differing numbers of hybridization events via an alignment strategy that matches events across matrices before computing distances. Our distance can also handle both isochronous and heterochronous networks, as well as timed networks.
%
We characterize mathematical properties of the proposed metrics and compare them to existing network-based distances on simulated and empirical datasets, showing that our metrics discriminate among alternative evolutionary scenarios, facilitate computation of summary statistics and construction of credible-sets for posterior network samples, and scale to realistic inference outputs such as posterior samples obtained with popular inference tools such as BEAST 2 \citep{Muller2020}.

\section{Definitions of network topologies}

All phylogenetic networks we consider are rooted and binary. We first assume that all tips correspond to samples obtained at the same time 0 (present time). We call this setting isochronous sampling. Figure \ref{fig:tree.res.ex} (A) shows an example of a rooted ranked (unlabeled) phylogenetic network. The formal definition is as follows:


\begin{definition} A rooted and ranked (unlabeled) phylogenetic network with $n$ leaves and $m$ hybridizations is a connected directed acyclic graph with increasing ordering of internal vertices and with the following characteristics: i) the root $r$ has indegree 0 and outdegree 2; ii) any leaf $v \in V_L$ has indegree 1 and outdegree 0; iii) any tree node $v \in V_T$ has indegree 1 and outdegree 2; iv) any hybrid node $v \in V_H$ has indegree 2 and outdegree 1; v) a tree edge $e \in E_T$ is an edge whose child is a tree node; vi) a hybrid edge $e \in E_H$ is an edge whose child is a hybrid node, and vii) a hybrid edge $e \in E_H$ has an inheritance probability parameter $\gamma_e < 1$ which represents the proportion of the genetic material that the child hybrid node received from this parent edge. \end{definition}

For the networks considered in this manuscript, we ignore the inheritance vector that represents the proportions of genetic material split at hybridization events.


\begin{definition} \label{def:labeled} A rooted and ranked (labeled) phylogenetic network on taxon set $X$ is a connected directed acyclic graph with vertices
$V = \{r\} \cup V_L \cup V_H \cup V_T$ , edges $E = E_H \cup E_T$ and equipped with a bijective leaf-labeling function $f : V_L \rightarrow X$.
\end{definition}

The class of networks of Definition \ref{def:labeled} is called \textit{explicit phylogenetic network} \citep{Huson2010}.
In a rooted (explicit or unlabeled) network, every internal node
represents a biological mechanism: speciation for tree nodes and
hybridization for hybrid nodes. However, other types of phylogenetic
networks also exist in the literature, such as unrooted networks
\citep{Huson2010} and semi-directed networks \citep{SolisLemus2016}.
Unrooted phylogenetic networks are typically obtained by suppressing
the root node and the direction of all edges.  In semi-directed
unrooted networks, on the other hand, the root node is suppressed and the direction of all tree edges are ignored, but the
direction of hybrid edges is maintained, thus preserving the information on which nodes are
hybrids. The placement of the root is then constrained, because the
direction of the two hybrid edges to a given hybrid node inform the
direction of time at this node: the third edge must be a tree edge
directed away from the hybrid node and leading to all the hybrid’s
descendants. Therefore the root cannot be placed on any descendant of
any hybrid node, although it might be placed on some hybrid edges.
Finally, implicit networks \citep{Huson2010} – also called split
networks - describe the discrepancy in gene trees, but they lack
biological interpretation as the internal nodes do not represent
ancestral species.

When phylogenetic networks are inferred in a two-step procedure in which gene trees are first estimated from multi-locus data, and then used to estimate the species network under the Multispecies Coalescent Model, different branches of the resulting network are allowed to be measured in different units of time (or scales) \citep{Meng2009, Yu2012,SolisLemus2016}. 
In contrast, inference
methods that integrate over gene trees (or co-estimate them) starting
from multi-locus sequences \citep{Zhang2018} are able to estimate
time-calibrated networks. 
Here, we focus on time-calibrated networks that are either ultrametric networks (isochronous sampling) or not, under heterochronous sampling. The latter is relevant in the context of rapidly evolving pathogens or analyses of ancient DNA samples.


\section{Unique encoding of unlabeled ranked phylogenetic networks as $\mathbf{F}$-matrices}

We extend the $\mathbf{F}$-matrix encoding of ranked tree shapes introduced in \citet{Kim2020} to unlabeled ranked  phylogenetic networks as follows. 
An $\mathbf{F}$-matrix encoding of a ranked and unlabeled phylogenetic network  with $n$ leaves and $m$ hybridization events is an $(n+2m) \times (n+2m)$ lower triangular matrix of non-negative integers $F_{i,j} \in \{0,1,2,\ldots, n+m\}$. For a given ranked and unlabeled phylogenetic network with $n$ leaves, starting from the root to the tips, we denote the time of a branching or hybridization event at each node $i$ by $u_{i}$ and the time interval between two consecutive nodes $i-1$ and $i$ by $I_{i}=(u_{i},u_{i-1})$. The diagonal elements of the $\mathbf{F}$-matrix indicate the number of branches at each time interval. The non-diagonal element $F_{i,j}$ represents the number of branches extant at $I_{j}=(u_{j},u_{j-1})$ and that are not involved in any event (birth nor hybridization) during the entire time interval $(u_{i},u_{j-1})$. An example of a phylogenetic network with $n=4$ taxa and its corresponding $\mathbf{F}$-matrix representation are shown in Figure \ref{fig:tree.res.ex}. The first column of the matrix is the vector $(1,0,\ldots,0)$ because there is one branch during $(u_{1},u_{0})$ but that branch bifurcates at time $u_{1}$. The second column is the vector $(2,1,0,\ldots,0)$ because there are two branches during $(u_{2},u_{1})$, one of them bifurcates at $u_{2}$ and the second one bifurcates at $u_{3}$. The fifth column is the vector $(5,3)$ because there are 5 branches during $(u_{5},u_{4})$ and two of them for a hybridization at time $u_{4}$. In the following, we provide a formal definition of the space of $\mathbf{F}$-matrices independent of the notion of ranked phylogenetic networks.

\begin{figure}[H]
\begin{subfigure}{0.48\textwidth}
\centering
\includegraphics[width = 0.7 \textwidth]{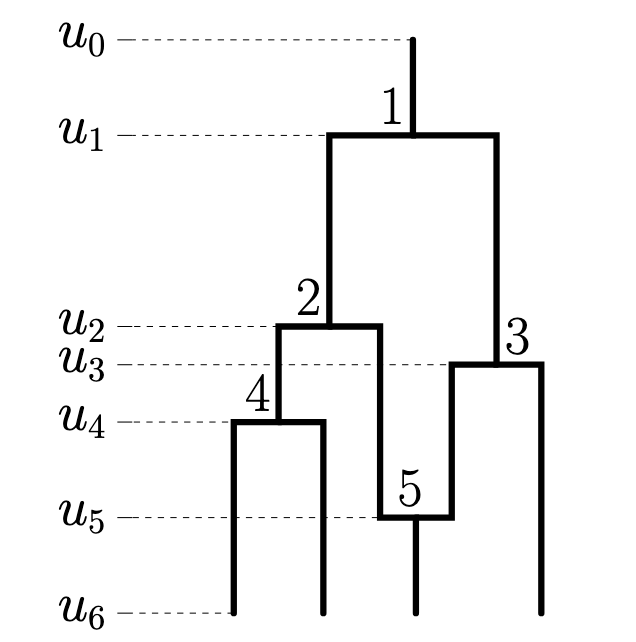}
\caption{}
\label{fig:F-matrix_network}
\end{subfigure}
\begin{subfigure}{0.45\textwidth}
\centering
\includegraphics[width = \textwidth]{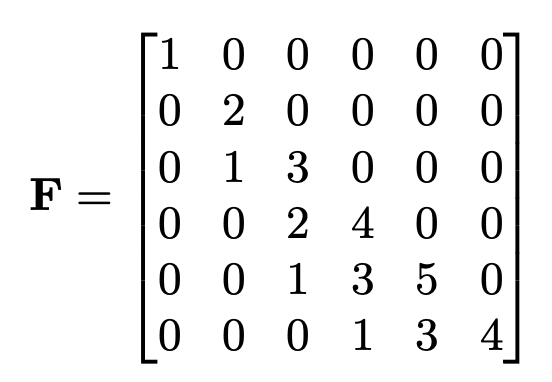}
\caption{}
\label{fig:F-matrix}
\end{subfigure}
\caption{Example of a phylogenetic network and its
    corresponding $\mathbf{F}$-matrix representation. (A) Change points occur at birth and hybridization times denoted by
  $u_{1},\ldots,u_{n+2m-1}$ irrespectively of the type of event and in increasing order from the root (at time $u_{0}$) to the
  tips (at time $u_{n+2m}$). (B) $F_{ij}$ denotes the number of branches that exist in
  $(u_{j},u_{j-1})$ and that do not give birth nor hybridize during the time
  interval $(u_{i},u_{j-1})$. 
  }
\label{fig:tree.res.ex}
\end{figure}

\begin{definition}\label{def:F1space} ($\mathbf{F}$-matrix). An $\mathbf{F}-$matrix $\mathbf{F} \in \mathcal{F}_{n+2m}$ of size $n+2m$ is a lower triangular square matrix of integers $F_{i,j} \in \{0,1,\ldots,n+m\}$ for all $i,j \in \{1,\ldots,n+2m\}$ satisfying the following properties:
\begin{enumerate}
\item The first column and row are $(1,0,\dots,0)$, i.e. $F_{1,1}=1$ and $F_{1,k}=F_{k,1}=0$ for $k\geq 2$.
\item The diagonal elements are $F_{1,1}=1$, $F_{2,2}=2, F_{3,3}=3$, and $F_{n+2m,n+2m}=n$,
$F_{i,i}>0$ and $F_{i,i} \in \{F_{i-1,i-1}-1,F_{i-1,i-1}+1\}$.
\item The subdiagonal elements are $F_{i+1,i}=F_{i,i}-1$ if $F_{i+1,i+1}-F_{i,i}=1$, and $F_{i+1,i}=F_{i,i}-2$ if $F_{i+1,i+1}-F_{i,i}=-1$ for $i \geq 2$.
\item The elements $F_{i,2}$ in the second column, we set $F_{3,2} = 1$. For $i=4,\ldots,n+2m$, there exists an index $k$ with $3 \leq k \leq n+2m$ such that $F_{i,2} = 1$ for $3 \leq i \leq k$ and $F_{i,2} = 0$ for $k < i \leq n+2m$.
\item All the other elements $F_{i,k},$ for $i=k+1,\ldots,n+2m$ and $k=3,\ldots,n+2m-1$ satisfy the following inequality $\max\{0,F_{i-1,k}-2,F_{i,k-1},F_{i,k-1}+F_{i-1,k}-F_{i-1,k-1}-2 \} \leq F_{i,k} \leq \min\{ F_{i-1,k},F_{i,k-1}+2,F_{i,k-1}+F_{i-1,k}-F_{i-1,k-1}\}.$

\end{enumerate}
\end{definition}

To see that an $\mathbf{F}$-matrix encodes a phylogenetic network, we relate the properties of Definition~\ref{def:F1space} to the phylogenetic network as follows. Property 1 states that the first column and first row are the vector $(1,0,\dots,0)$. We start with a single lineage at the root, followed by a bifurcation at time $u_1$. Property 2 states that we only consider phylogenetic networks that start with two consecutive speciation events ($F_{2,2}=2, F_{3,3}=3$). This is consistent with network identifiability \citep{hector2017network}.  In addition, property 2 states that  $F_{i+1,i+1}=F_{i,i}+1$ whenever a bifurcation occurs at time $u_i$ because one lineage splits into two lineages, and $F_{i+1,i+1}=F_{i,i}-1$ whenever a hybridization occurs at time $u_i$ because two lineages merge into one lineage. Since the diagonal entries begin at 1 and the $(n+2m)$-th diagonal entry is equal to n, it follows that there are $n+m$ events of bifurcation ($F_{i+1,i+1}=F_{i,i}+1$) and $m$ events of hybridization ($F_{i+1,i+1}=F_{i,i}-1$). Indeed, $\sum^{n+2m-1}_{i=1} 1\{F_{i,i}-F_{i+1,i+1}>0\}=m$, which corresponds to the total number of hybridization events. In addition, property 2 also implies that $F_{i,i}\leq n+m-\sum^{i-1}_{j=2} \mathds{1}(F_{j,j}-F_{j-1,j-1}=-1)$. To prove it, suppose $F_{i,i} > n+m-\sum^{i-1}_{j=2} \mathds{1}(F_{j,j}-F_{j-1,j-1}=-1)$ for some $i$. Then $F_{n+2m, n+2m} \geq F_{i,i} - (m-\sum^{i-1}_{j=2} \mathds{1}(F_{j,j}-F_{j-1,j-1}=-1)) > n+m-\sum^{i-1}_{j=2} \mathds{1}(F_{j,j}-F_{j-1,j-1}=-1) - (m-\sum^{i-1}_{j=2} \mathds{1}(F_{j,j}-F_{j-1,j-1}=-1)) = n$, which is a contradiction.
Property 3 states that if the $i$-th event is a bifurcation, then one of lineages in the time interval $(u_i, u_{i-1})$ bifurcates at time $u_i$, which means that in column $i$ the value in row $i+1$ is one less than that in row $i$; whereas it decreases by two if the $i$-th event is a hybridization. 
Property 4 refers to the second column of $\mathbf{F}$. It says that the two branches created at the root $F_{2,2}=2$, one of the branches bifurcates at $u_{2}$, i.e. $F_{3,2}=1$, then $F_{k,2}=1$ remains being 1 for $k>3$, or until the other branch bifurcates or hybridizes at $u_{k'}$ for some $k'$ with $3 \leq k' \leq n+2m$ and $F_{l,2}=0$ for $l>k'$.
Property 5 can be expressed as a series of inequality constraints on the values of $F_{i,k}$ given it's neighboring values $F_{i-1,k}$, $F_{i-1,k-1}$, $F_{i,k-1}$. The first inequality says that $\max\{0,F_{i-1,k}-2\} \leq F_{i,k} \leq F_{i-1,k}$; that is, columns are non-increasing and the $i$-th value of column $k$ will be either the previous row value of the same column or the previous row value of the same column minus one or two for rows $i \geq k+1$. The second inequality says that $F_{i,k-1} \leq F_{i,k} \leq F_{i,k-1}+2$, that is rows are non-decreasing and the $k$-th value of row $i$ will be either the previous column value of the same row or the previous column value of the same raw plus one or two for column $k \leq i-1$. The last inequality says that $F_{i-1,k}-F_{i-1,k-1}-2 \leq F_{i,k}-F_{i,k-1} \leq F_{i-1,k}-F_{i-1,k-1}$, i.e. the difference of consecutive values in the i-th row between the $k-1$ and the $k$ columns is either the same as the difference of consecutive values in the previous row or the difference minus one or two. 


Definition \ref{def:F1space} provides an algorithm for enumerating  all possible $\mathbf{F}$-matrices and hence enumerate all possible ranked unlabeled phylogenetic networks. 
To enumerate all $\mathbf{F}$ matrices, we can start by enumerating all possible diagonal vectors (Property 3), together with the constrains that $F_{n+2m,n+2m}=n$, $F_{i,i}>1$ for $i>1$, and  $\sum^{n+2m-1}_{i=1} 1\{F_{i,i}-F_{i+1,i+1}>0\}=m$, and proceed sequentially row by row starting at the fourth row and moving from left to right within the rows.
For example, the $\mathbf{F}$-matrices with $n=3, m=1$ are of the following general form:
\begin{equation}
\label{eq:mtc}
\mathbf{F}=\begin{bmatrix}
1 & 0 & 0 & 0 & 0\\
0 & 2 & 0 & 0 & 0\\
0 & 1&3 &0 & 0 \\
0 & F_{3,1} & 1 & 2 & 0 \\
0 & F_{4,1}& F_{4,2} & 1 & 3 \\
\end{bmatrix}
\& \text{    }
\mathbf{F}=\begin{bmatrix}
1 & 0 & 0 & 0 & 0\\
0 & 2 & 0 & 0 & 0\\
0 & 1&3 &0 & 0 \\
0 & F_{3,1} & 2 & 4 & 0 \\
0 & F_{4,1}& F_{4,2} & 2 & 3 \\
\end{bmatrix}
\end{equation}
where $F_{3,1}$ and $F_{4,1}$ satisfy Property 4 and $F_{4,2}$ satisfies Property 5 in Definition \ref{def:F1space}. There are $4$ different $\mathbf{F}$-matrices of $n=3$ and $m=1$ of the general form of the left matrix in (\ref{eq:mtc}):
 \begin{equation*}
\mathbf{F}_{1}=\begin{bmatrix}
1 & 0 & 0 & 0 & 0\\
0 & 2 & 0 & 0 & 0\\
0 & 1&3 &0 & 0\\
0 & 0 & 1 & 2 & 0 \\
0 & 0& 0 & 1 & 3\\
\end{bmatrix}
,\mathbf{F}_{2}=\begin{bmatrix}
1 & 0 & 0 & 0 & 0\\
0 & 2 & 0 & 0 & 0\\
0 & 1&3 &0 & 0\\
0 & 1 & 1 & 2 & 0 \\
0 & 0& 0 & 1 & 3\\
\end{bmatrix}
,\mathbf{F}_{3}=\begin{bmatrix}
1 & 0 & 0 & 0 & 0\\
0 & 2 & 0 & 0 & 0\\
0 & 1&3 &0 & 0\\
0 & 1 & 1 & 2 & 0 \\
0 & 1& 1 & 1 & 3\\
\end{bmatrix}
,\mathbf{F}_{4}=\begin{bmatrix}
1 & 0 & 0 & 0 & 0\\
0 & 2 & 0 & 0 & 0\\
0 & 1&3 &0 & 0\\
0 & 0 & 1 & 2 & 0 \\
0 & 0& 1 & 1 & 3\\
\end{bmatrix}
\end{equation*}
and 7 different $\mathbf{F}$-matrices of $n=3$ and $m=1$ of the form of the right matrix in (\ref{eq:mtc}):
\begin{equation*}
\mathbf{F}_{5}=\begin{bmatrix}
1 & 0 & 0 & 0 & 0\\
0 & 2 & 0 & 0 & 0\\
0 & 1&3 &0 & 0\\
0 & 0 & 2 & 4 & 0 \\
0 & 0& 1 & 2 & 3\\
\end{bmatrix}
,\mathbf{F}_{6}=\begin{bmatrix}
1 & 0 & 0 & 0 & 0\\
0 & 2 & 0 & 0 & 0\\
0 & 1&3 &0 & 0\\
0 & 0 & 2 & 4 & 0 \\
0 & 0& 2 & 2 & 3\\
\end{bmatrix}
,\mathbf{F}_{7}=\begin{bmatrix}
1 & 0 & 0 & 0 & 0\\
0 & 2 & 0 & 0 & 0\\
0 & 1&3 &0 & 0\\
0 & 1 & 2 & 4 & 0 \\
0 & 0& 1 & 2 & 3\\
\end{bmatrix}
\end{equation*}
\begin{equation*}
\mathbf{F}_{8}=\begin{bmatrix}
1 & 0 & 0 & 0 & 0\\
0 & 2 & 0 & 0 & 0\\
0 & 1&3 &0 & 0\\
0 & 1 & 2 & 4 & 0 \\
0 & 0& 0 & 2 & 3\\
\end{bmatrix}
,\mathbf{F}_{9}=\begin{bmatrix}
1 & 0 & 0 & 0 & 0\\
0 & 2 & 0 & 0 & 0\\
0 & 1&3 &0 & 0\\
0 & 1 & 2 & 4 & 0 \\
0 & 1& 1 & 2 & 3\\
\end{bmatrix}
,\mathbf{F}_{10}=\begin{bmatrix}
1 & 0 & 0 & 0 & 0\\
0 & 2 & 0 & 0 & 0\\
0 & 1&3 &0 & 0\\
0 & 1 & 2 & 4 & 0 \\
0 & 1 & 2 & 2 & 3\\
\end{bmatrix}
,\mathbf{F}_{11}=\begin{bmatrix}
1 & 0 & 0 & 0 & 0\\
0 & 2 & 0 & 0 & 0\\
0 & 1 & 3 & 0 & 0\\
0 & 0 & 2 & 4 & 0 \\
0 & 0 & 0 & 2 & 3\\
\end{bmatrix}
\end{equation*}

Figure \ref{fig:11nets} shows the 11 unlabeled networks with $n=3$ tips and $m=1$ hybridization event corresponding to the 11 $\textbf{F}$ matrices (internal node labels (rankings) are suppressed for clarity in the depictions). 
\begin{figure}[H]
    \centering
    \includegraphics[scale=0.2]{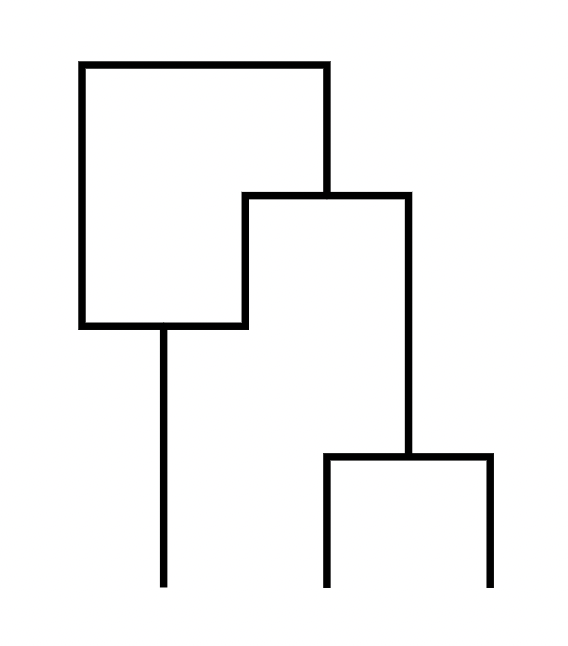}
    \includegraphics[scale=0.2]{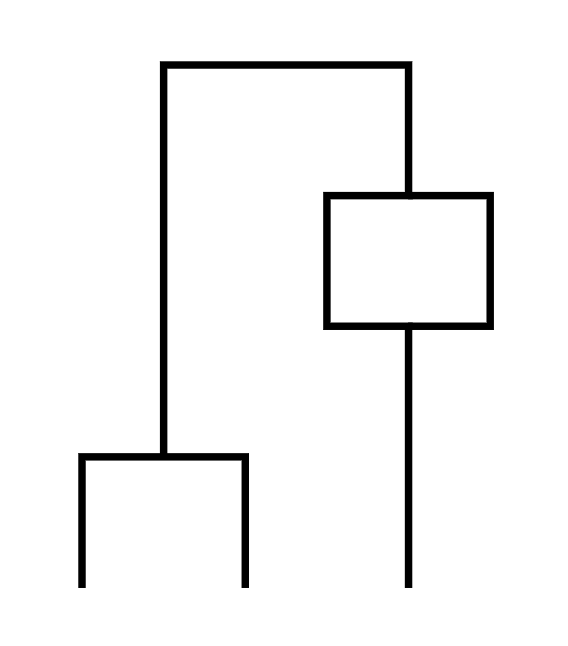}
    \includegraphics[scale=0.2]{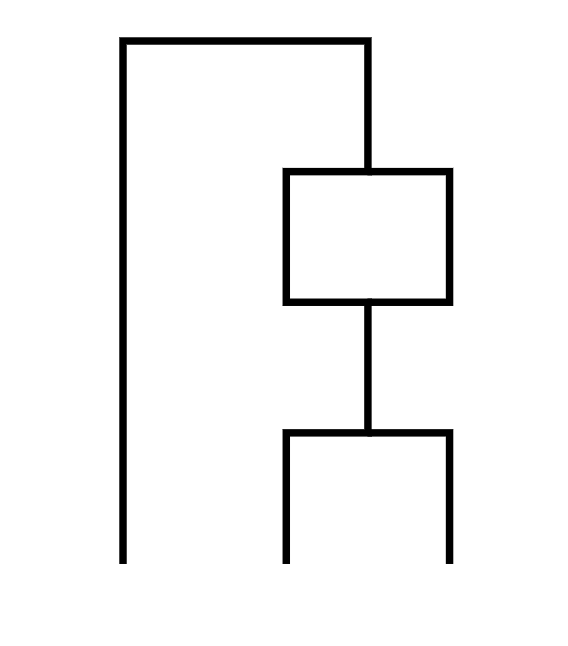}
    \includegraphics[scale=0.2]{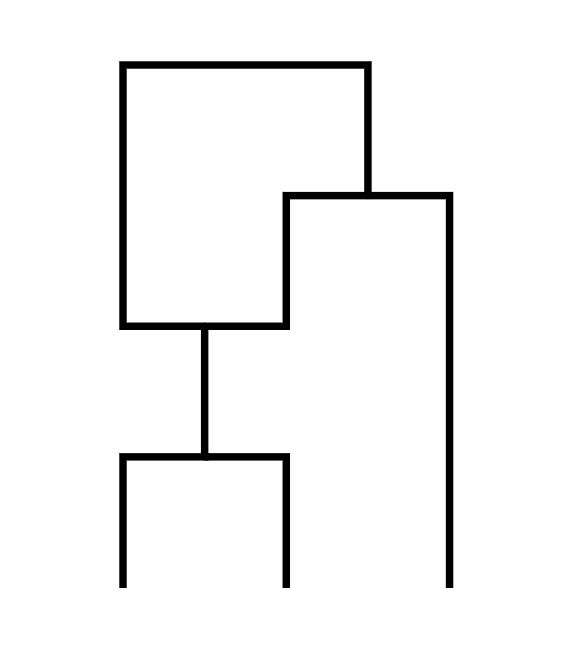}\\
    \includegraphics[scale=0.2]{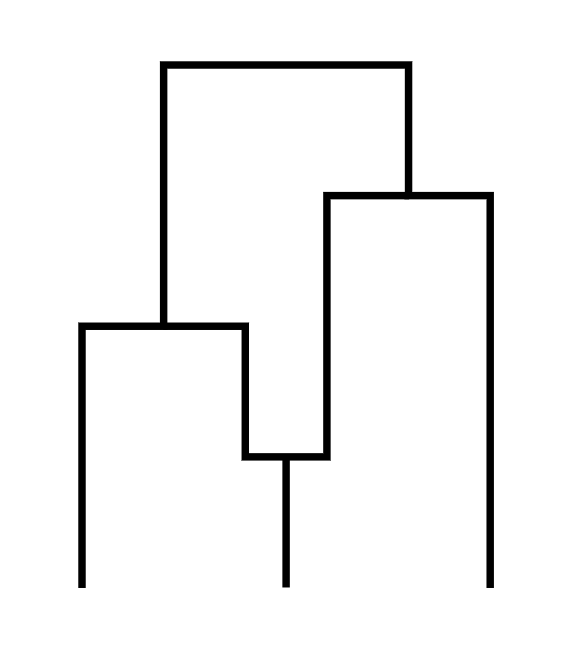}
    \includegraphics[scale=0.2]{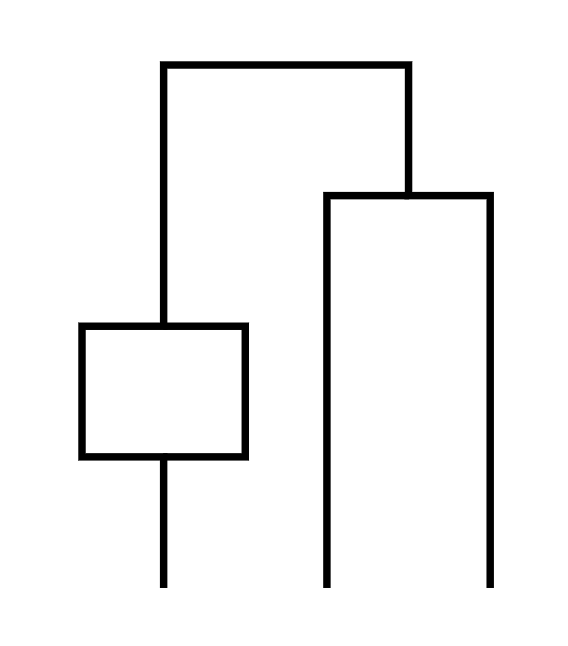}
    \includegraphics[scale=0.2]{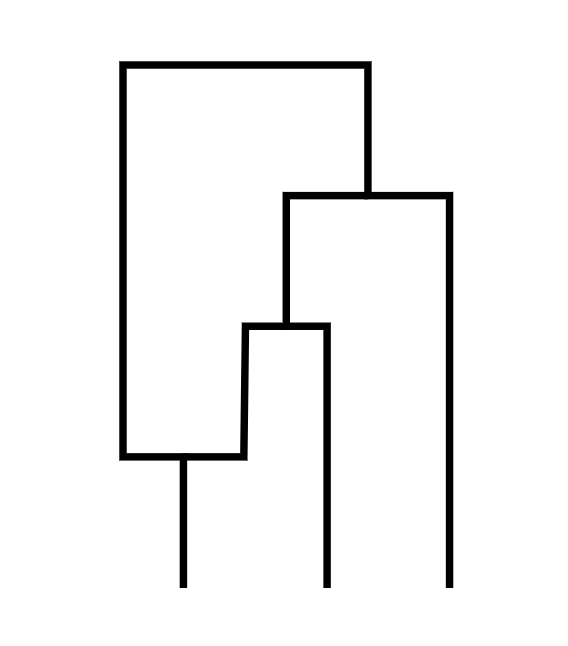}
    \includegraphics[scale=0.2]{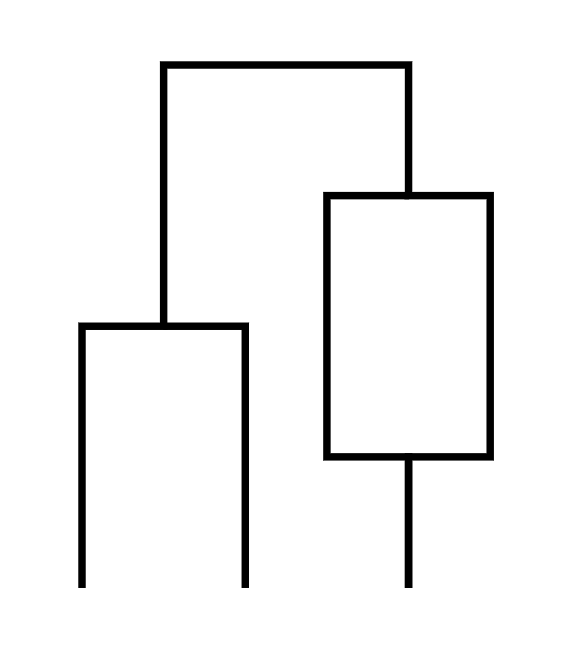}
    \includegraphics[scale=0.2]{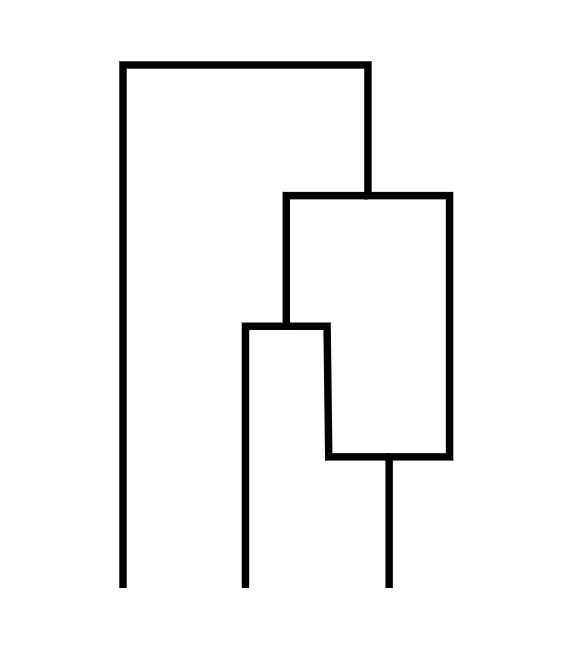}
    \includegraphics[scale=0.2]{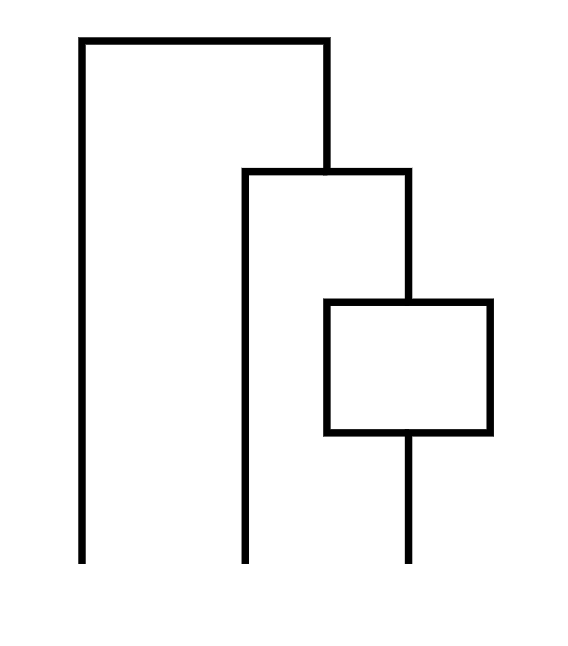}
    \includegraphics[scale=0.2]{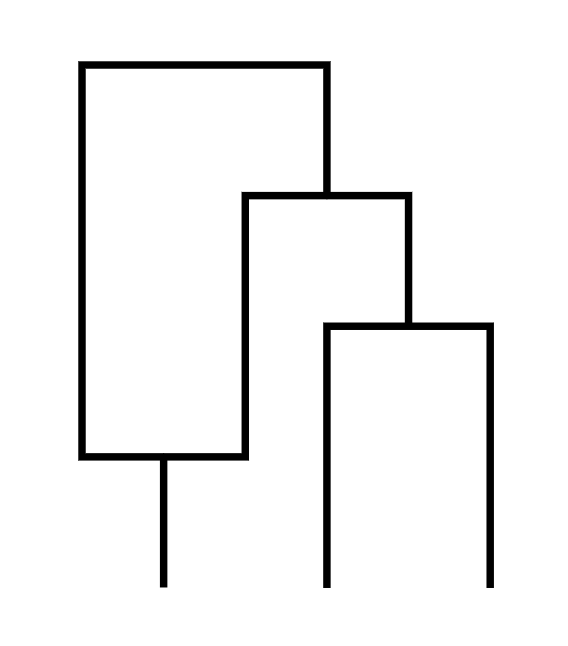}
    \caption{Complete list of unlabeled ranked phylogenetic networks with $n=3$ leaves and $m=1$ hybridization event corresponding to the matrices $\textbf{F}_1,...,\textbf{F}_{11}$ from left to right respectively (top: general form of the left matrix in (\ref{eq:mtc}); bottom: general form of the right matrix in (\ref{eq:mtc})).}
    \label{fig:11nets}
\end{figure}

%
%
%
%
%
%

\subsection{Bijective map between unlabeled ranked networks and $\mathbf{F}$ matrices}


We note that the space of $\mathbf{F}$-matrices is defined independently of the notion of ranked phylogenetic networks. The next theorem establishes a bijection between the two spaces.

\begin{theorem} \textbf{(Bijection).}\label{prop:bij}
There exists a bijective mapping between the space of rooted ranked and unlabeled phylogenetic network topologies of $n$ leaves and $m$ hybridizations $\mathcal{P}^{R}_{n,m}$ and the space $\mathcal{F}_{n,m}$ of $\mathbf{F}$-matrices of size $n+2m.$
\end{theorem}

\begin{proof}
Consider a rooted and ranked network shape $P_{n,m} \in \mathcal{P}^{R}_{n,m}$ with $n$ leaves and $m$ times of hybridization. It is a rooted, directed and acyclic graph in which hybridization nodes have  in-degree 2 and out-degree 1, speciation nodes have in-degree 1 and out-degree 2, leaves are nodes 
with in-degree 1 and out-degree 0, and finally the unique node with in-degree 0 and out-degree
1, which is called the “root”. There is a total of $e=n+2m$ internal nodes (hybridization and speciation nodes, and the root). The internal nodes of $P_{n,m}$ are uniquely labeled by the indices (or rankings) of their event times, and all leaves of $P_{n,m}$ share the same label $e=n+2m$. 

Now, the root is the only node with in-degree 0 and out-degree 1. This root is uniquely defined in the $\mathbf{F}$-matrix by a column with 1 in the first entry and 0s everywhere else, so we can proceed by demonstrating the bijection ignoring the root node. 

We will define a mapping between $\mathcal{P}^{R}_{n,m}$ and the space $\mathcal{F}_{n,m}$ by first using graph notation and show that the graph $\mathbf X$ completely specifies $P_{n,m}$. We  will then construct an explicit mapping from $\mathbf X$ to $\mathbf F$ and show that this mapping is bijective.

We define 
$I=\{1,\ldots,n+2m-1\}$ to be a set of all internal nodes (excluding the root), and all leaf nodes to be $n+2m$. For a speciation node $i \in I$, let $o_i = (x_{i,1}, x_{i,2})$ indicate the ordered pair of labels of the two immediate descendants of speciation node $i$ such that $i < x_{i,1} \leq x_{i,2}$. 
Similarly, for a hybridization node $i\in I$, we set $o_i = (x_{i,1},x_{i,2})$, where we set $x_{i,1}=0$ and  $x_{i,2}$ correspond to the label of the immediate descendant of hybridization node $i$ such that $i < x_{i,2}$.  We denote the set of all pairs of $i$ and $o_i=(x_{i,1}, x_{i,2})$ in $P_{n,m}$ by $\mathbf X=\{(i, o_i) \mid i \in I \}$. Then $\mathbf X$ and the root completely specifies $P_{n,m}$. Figure \ref{fig:bijection_figure} shows an example of a network $P_{5,1}$ with 5 leaves and 1 hybridization. The internal node with index 4 is a speciation node (Figure \ref{fig:bijection_figure1}), with  descendant nodes 6 and 7, therefore $x_{4,1}=6$ and $x_{4,2}=7$, and so we set $o_4 = (6,7)$ and $(4, o_4) = (4,(6,7))$. In Figure \ref{fig:bijection_figure2}, internal node 3 is a hybridization node with  immediate descendant node 4. Therefore, we set $o_3 = (0,4)$ and $(3, o_3) = (3,(0,4))$.

\begin{figure}[H]
\begin{subfigure}[b]{0.45\textwidth}
\centering
\includegraphics[width = \textwidth]{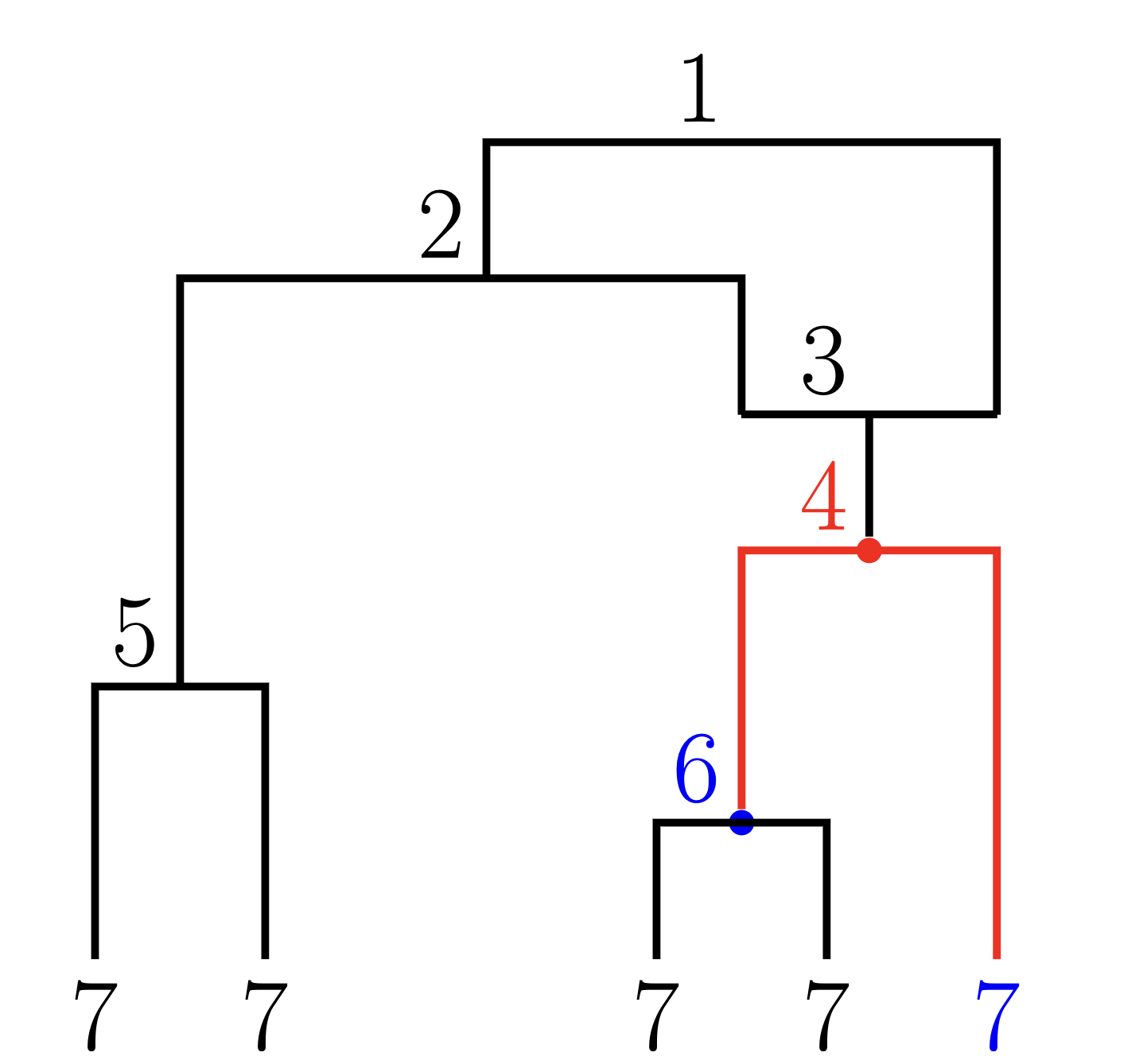}
\caption{}
\label{fig:bijection_figure1}
\end{subfigure}
\hspace{0.09\textwidth}
\begin{subfigure}[b]{0.45\textwidth}
\centering
\includegraphics[width = \textwidth]{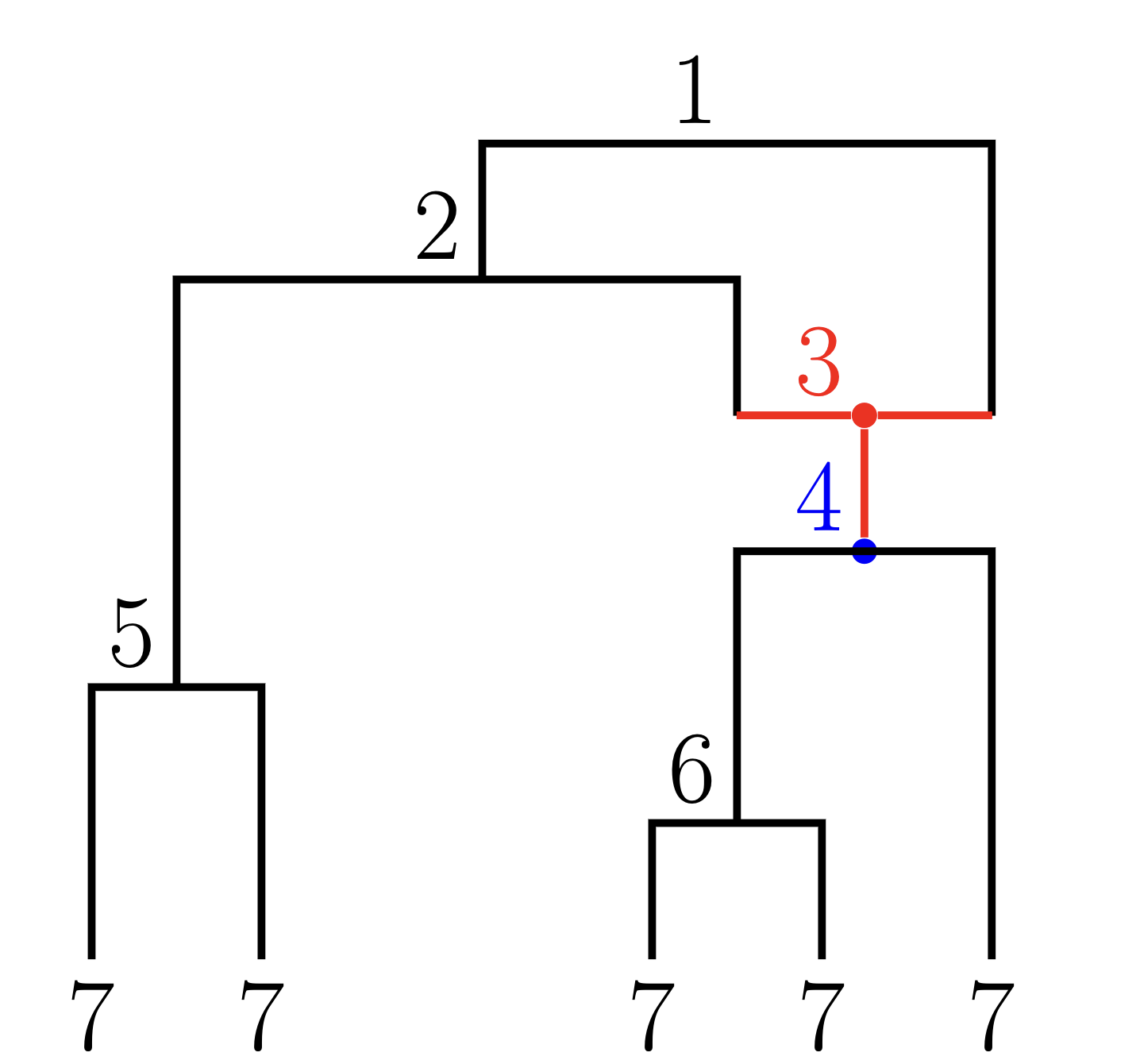}
\caption{}
\label{fig:bijection_figure2}
\end{subfigure}
\caption{A ranked phylogenetic network  exemplifying notation used in the proof of Theorem \ref{prop:bij}. In (A), speciation node with label 4 (red) has two descendants $6$ and $7$ (blue), this information is encoded in graph notation as $(4,o_{4})=(4,(6,7))$. In (B) hybridization node 3 (red) has one descendant node with label 4 (blue). This information is encoded in graph notation as $(3,o_{3})=(3,(0,4))$. Network $\mathbf{X}$ in graph notation is encoded as  $\{(1, (2,3)), (2,(3,5)), (3,(0,4)), (4,(6,7)), (5,(7,7)), (6,(7,7))\}$.
}
\label{fig:bijection_figure}
\end{figure}

We define an auxiliary matrix $\mathbf D$ to be a $(n+2m-1) \times (n+2m-1)$ lower triangular matrix as follows
\begin{align*}
D_{i,j}=\phi(i,o_{i})_{j}=\begin{cases}
0 & \text{ if } 1 \leq j <i \text{ or } x_{i,2} \leq j < e.\\
1 & \text{ if } x_{i,1} \leq j < x_{i,2} \\
2 & \text{ if } i \leq j < x_{i,1} \end{cases}
\end{align*}

where $\phi: \mathbf X \rightarrow \{0,1,2\}^{e-1}$.The $j$-th element of $\phi(i, o_i)$ is the number of immediate descendants of internal node $i$ present or extant at the time interval $(u_{j+1}, u_{j})$. $\phi$ is an injective map. To prove this, consider $i, j \in I$ and $i \neq j$; without loss of generality, assume $i < j$. Then, the $i$-th element of $\phi(i,o_{i})$ is 1 or 2 while the $i$-th element of $\phi(j,o_{j})$ is 0, and thus $\phi(i,o_{i}) \neq \phi(j,o_{j})$.

$D_{i,j}$ must have the following properties:
\begin{itemize}
 \item[$D1:$] $D_{ii} = 2$ if the $i$-th node is a speciation node,
    and $D_{ii} = 1$ if it  is a hybridization node ($D_{11} = 2$ and $D_{22} = 2$).
    \item[$D2:$] $D_{(i-1)j}-2 \leq D_{ij} \leq D_{(i-1)j}$ for $i=j+1,\ldots,n+2m-1$, and $j=2,\ldots,n+2m-2$.
    \item[$D3:$] If node $i$ ($i>1$) is a  hybridization node, then $\sum_{j=1}^{i-1} \mathbf{1}_{\{D_{(i-1)j} - D_{(i)j} = 1\}} + 2 \times \mathbf{1}_{\{D_{(i-1)j} - D_{(i)j} = 2\}} = 2$.
    If node $i$ ($i>1$) is a speciation node, then $\sum_{j=1}^{i-1} \mathbf{1}_{\{D_{(i-1)j} - D_{ij} = 1\}} = 1$.
\end{itemize}

Thus, this matrix $\mathbf D$ is uniquely determined for any given rooted, ranked, and unlabeled binary network shape $P$ in $\mathcal{P}_{n,m}^R$. $\mathbf D$ is a lower triangular matrix. Therefore, it is sufficient to prove that the map $\varphi: \mathcal{D}_{n,m} \rightarrow \mathcal{F}_{n,m}$ is a surjective mapping, where $\varphi(\mathbf D) = \mathbf F$ be given by $F_{ij} = \sum_{k=1}^{j-1} D_{(i-1)k}$ for $i\geq 2$ and $j \leq i$. 
The inverse mapping is $\varphi^{-1}: \mathcal{F}_{n,m} \rightarrow \mathcal{D}_{n,m}$ which is given by $D_{ij} = F_{(i+1)(j+1)} - F_{(i+1)j}$ for $1\leq j \leq i \leq n+2m-1$.

Again, an $\mathbf{F}$-matrix of size $n+2m$ in the space $\mathcal{F}_{n,m}$ is a lower triangular matrix satisfying

\begin{itemize}
\item[$F1:$] The first column and row are $(1,0,\dots,0)$, i.e. $F_{1,1}=1$ and $F_{1,k}=F_{k,1}=0$ for $k\geq 2$.
\item[$F2:$] The diagonal elements are $F_{1,1}=1$, $F_{2,2}=2, F_{3,3}=3$, and $F_{n+2m,n+2m}=n$,
$F_{i,i}>0$ and $F_{i,i} \in \{F_{i-1,i-1}-1,F_{i-1,i-1}+1\}$.
\item[$F3:$] The subdiagonal elements are $F_{i+1,i}=F_{i,i}-1$ if $F_{i+1,i+1}-F_{i,i}=1$, and $F_{i+1,i}=F_{i,i}-2$ if $F_{i+1,i+1}-F_{i,i}=-1$ for $i \geq 2$.
\item[$F4:$] The elements $F_{i,2}$ in the second column, we set $F_{3,2} = 1$. For $i=4,\ldots,n+2m$, there exists an index $k$ with $3 \leq k \leq n+2m$ such that $F_{i,2} = 1$ for $3 \leq i \leq k$ and $F_{i,2} = 0$ for $k < i \leq n+2m$.
\item[$F5:$] All the other elements $F_{i,k}$, for $i=k+1,\ldots,n+2m$ and $k=3,\ldots,n+2m-1$ satisfy the following inequality $\max\{0,F_{i-1,k}-2,F_{i,k-1},F_{i,k-1}+F_{i-1,k}-F_{i-1,k-1}-2 \} \leq F_{i,k} \leq \min\{ F_{i-1,k},F_{i,k-1}+2,F_{i,k-1}+F_{i-1,k}-F_{i-1,k-1}\}.$
\end{itemize}

Let $\mathbf F$ be an element of $\mathcal{F}_{n,m}$ and $\mathbf D = \varphi^{-1}(\mathbf F).$ To prove $\varphi$ is a surjective mapping, it is sufficient to prove that $\mathbf D$ satisfies $D1$ to $D3$.

We first see that $D1$ is automatically satisfied by the definition of $\varphi^{-1}$: If the $i$-th node is a speciation node, then
$D_{ii} = F_{(i+1)(i+1)} - F_{(i+1)i} = F_{(i+1)(i+1)} - (F_{ii}-1) = F_{(i+1)(i+1)} - F_{ii}+1= 2$ using $F2$ and $F3$ ; and if $i$ is a hybridization node, then $D_{ii} = F_{(i+1)(i+1)} - F_{(i+1)i} = F_{(i+1)(i+1)} - (F_{ii}-2) = 1-2 = 1$ using $F2$ and $F3$. 

To prove that $\mathbf D$ satisfies $D2$, the inequality in $F4$ shows that the first column of $\mathbf D$ satisfies $D2$ by the definition of $\varphi^{-1}$. In addition, the inequality in $F5$, where $F_{i,k-1}+F_{i-1,k}-F_{i-1,k-1}-2 \leq F_{i,k} \leq F_{i,k-1}+F_{i-1,k}-F_{i-1,k-1}$ for $i=k+1,\ldots,n+2m$, and $k=3,\ldots,n+2m-1$ shows that $\mathbf D$ satisfies $D2$ for the rest columns by the definition of $\varphi^{-1}$. This is because it implies $(F_{i-1,k}-F_{i-1,k-1})-2 \leq F_{i,k} - F_{i,k-1} \leq F_{i-1,k}-F_{i-1,k-1}$, which is $D_{(i-2)(k-1)}-2 \leq D_{(i-1)(k-1)} \leq D_{(i-2)(k-1)}$ for $i=k+1,\ldots,n+2m$, and $k=3,\ldots,n+2m-1$. This implies $D_{(i-1)j}-2 \leq D_{ij} \leq D_{(i-1)j}$ for $i=j+1,\ldots,n+2m-1$, and $j=2,\ldots,n+2m-2$.

Lastly, inequalities in $F5$ imply that, for $i=1,\ldots,n+2m-1$ and $j=1,\ldots,i$, $D_{ij}-D_{(i+1)j}$ equals either 1 or 2. According to $F2$ and $F3$ and the definition of $\varphi$, $\sum_{j=1}^{i-1} D_{(i-1)j} - \sum_{j=1}^{i-1} D_{ij} = F_{ii} - F_{(i+1)i}$ equals 1 if $F_{(i+1)(i+1)} - F_{ii} = 1;$ or equals 2 if $F_{(i+1)(i+1)} - F_{ii} = -1$ for $i \geq 2$. By combining the two conclusions above, we prove that $\mathbf D$ satisfies $D3$.
   
\end{proof}

\section{Distance metrics on phylogenetic networks}

\subsection{Metrics on rooted unlabeled ranked phylogenetic networks} 

Having established the bijection between the space of $\mathbf{F}$-matrices and rooted ranked and unlabeled phylogenetic network topologies, we  define two distance functions $d_{1}$  and $d_{2}$ on the space $\mathcal{P}_{n,m}$ of phylogenetic networks with $n$ leaves and $m$ hybridizations as the mappings $d_{i}: \mathcal{P}_{n,m} \times \mathcal{P}_{n,m} \rightarrow \mathbb{R}_+$, where $\mathbb{R}_+ = \{x \in \mathbb{R} | x \geq 0 \}$, $i=1,2$, such that for $\mathbf{F}^{(1)}_{n,m}, \mathbf{F}^{(2)}_{n,m} \in \mathcal{P}_{n,m}$,
    \begin{align*}
        d_1(\mathbf{F}^{(1)}_{n,m}, \mathbf{F}^{(2)}_{n,m}) &= \sum_{i,j} \abs*{F_{i,j}^{(1)} - F_{i,j}^{(2)}}, \\
        d_2(\mathbf{F}^{(1)}_{n,m}, \mathbf{F}^{(2)}_{n,m}) &= \sqrt{\sum_{i,j} \left(F_{i,j}^{(1)} - F_{i,j}^{(2)} \right)^2}.
    \end{align*}
The two distances are indeed metrics; both inherit the properties of  $L_{1}$ and $L_{2}$ (euclidean) norms. These distances are only defined on matrices representing phylogenetic networks with the same number of leaves and hybridization events. We allow for different number of hybridization events in Section \ref{sec:diff-hyb}. 

\subsection{Metrics on rooted unlabeled ranked timed phylogenetic networks with isochronous sampling times}
\label{sec:bl}

In the following definition, we include branch lengths to define distances between isochronous timed phylogenetic networks. We define two distances $d^{w}_{1}$ and $d^{w}_{2}$ on the space of ranked, unlabeled and timed phylogenetic networks with $n$ leaves sampled at the same time and $m$ hybridization events. We define the weight matrix $\mathbf W$ as a triangular matrix of size $n+2m$ that contains the relevant branch length information:
$W_{i,j}=u_{j-1}-u_{i}$ for $j<i$ and $W_{i,j}=0$ otherwise. For a pair of timed phylogenetic networks $\mathbf{g}^{(i)}_{n,m}$, $i=1,2$
    \begin{align*}
        d_1^{w}(\mathbf{g}^{(1)}_{n,m}, \mathbf{g}^{(2)}_{n,m}) &= \sum_{i,j} \abs*{F_{i,j}^{(1)}W^{(1)}_{i,j} - F_{i,j}^{(2)}W^{(2)}_{i,j}}, \\
        d_2^{w}(\mathbf{g}^{(1)}_{n,m}, \mathbf{g}^{(2)}_{n,m}) &= \sqrt{\sum_{i,j} \left(F_{i,j}^{(1)} W^{(1)}_{i,j}- F_{i,j}^{(2)}W^{(2)}_{i,j} \right)^2}.
    \end{align*}
where $\mathbf{W}^{(1)}$ and $\mathbf{W}^{(2)}$ are the weight matrices associated with $\mathbf{g}^{(1)}_{n,m}$ and $\mathbf{g}^{(2)}_{n,m}$ respectively. The product of  $F_{i,j}$, the number of extant branches not involved in any event in $(u_{i},u_{j-1})$, and $W_{i,j}$, the corresponding time $u_{j-1}-u_{i}$, is the total branch length of lineages not involved in any event. Therefore, $d^{w}_{1}$ and $d^{w}_{2}$ can be interpreted as weighted distances, weighted by evolutionary time.


\subsection{Metrics on rooted unlabeled ranked timed phylogenetic networks with heterochronous sampling times}
\label{sec:heterochronous}
In order to compute distances between timed phylogenentic networks with heterochronous sampling, we need to compute the branch length of lineages not involved in any event and substract the time period a lineage was not present due to sampling.
An $h\mathbf{F}$-matrix is an extension of $\mathbf F$-matrix for rooted, unlabeled, timed, binary and heterochronous phylogenetic networks with $n$ leaves and $m$ hybridization events, where leaves can have different collection times. The $h\mathbf{F}$-matrix is an $(n+2m) \times (n+2m)$ lower triangular matrix, in which $(hF)_{i,j}$ denotes the total branch length of lineages extant and that are not involved in any event (branching or hybridization)  in the interval $(u_{i}, u_{j-1})$. Let $x_{1},\ldots,x_{n}$ denote the sampling times of the $n$ samples, then to compute $(hF)_{i,j}$ we first compute the total branch length assuming the network is isochronous as before i.e $F_{i,j}(u_{j-1}-u_{i})$ and then substract the missing branch lengths in the interval  $\sum^{n}_{k=1}(x_{k}-u_{i})\mathbbm{1}(u_{i}<x_{k}<u_{j-1})$.

In the example illustrated in Figure \ref{fig:hF_example}, the network $H_{4,1}$ is a heterochronous phylogenetic network with four leaf nodes and a single hybridization event. When using the $h\mathbf{F}$-matrix embedding, for example, the entry $(hF)_{3,3}$ is $3(u_{2}-u_{3})- (x_{1}-u_{3})=3(u_2-x_1)+2(x_1-u_3)$ representing the sum of branch lengths within the interval $(u_{3}, u_{2})$. 

\begin{figure}[H]
\begin{subfigure}{0.45\textwidth}
\centering
\includegraphics[width = 0.8 \textwidth]{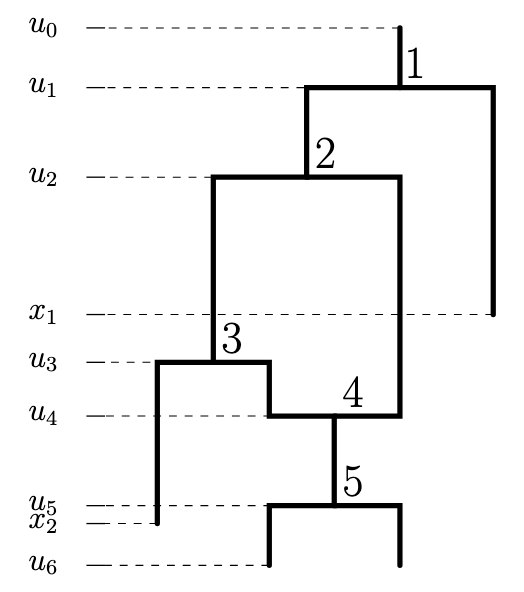}
\caption{}
\label{fig:heterochronous_network}
\end{subfigure}
\begin{subfigure}{0.48\textwidth}
\centering
\includegraphics[width =  \textwidth]{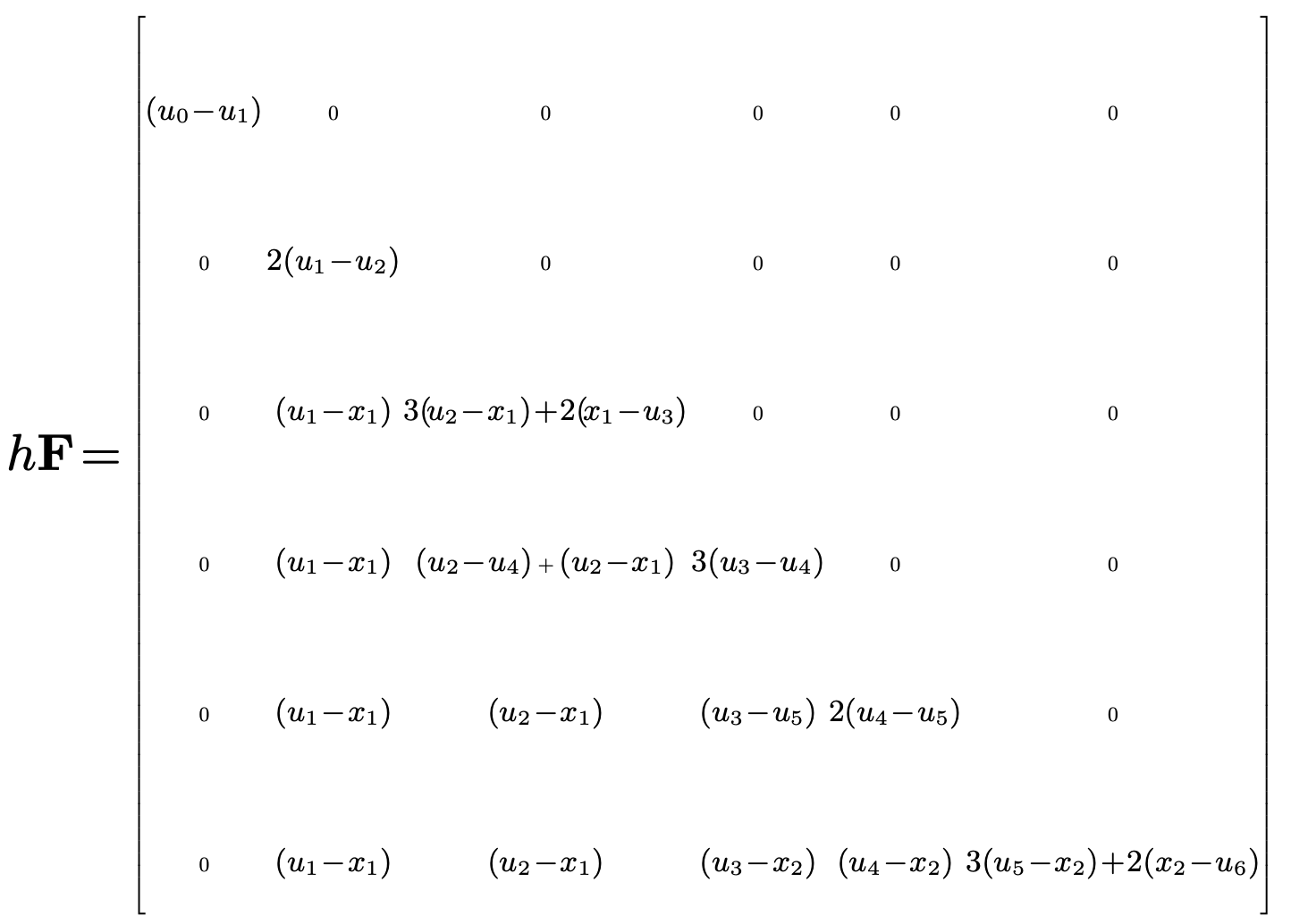}
\caption{}
\label{fig:hF_matrix}
\end{subfigure}
\caption{(A) A heterochronous phylogenetic network $H_{4,1}$. (B) $h\mathbf{F}$-matrix embedding of the network.}
\label{fig:hF_example}
\end{figure}

Based on the $h\mathbf{F}$-matrix embedding, we define two distance functions $d_1^{h}$ and $d_2^{h}$ for heterochronous phylogenetic networks with $n$ leaves and $m$ hybridizations. For a pair of heterochronous phylogenetic networks $\mathbf{h}^{(i)}_{n,m}$, $i=1,2$
    \begin{align*}
        d_1^{h}(\mathbf{h}^{(1)}_{n,m}, \mathbf{h}^{(2)}_{n,m}) &= \sum_{i,j} \abs*{(hF)_{i,j}^{(1)} - (hF)_{i,j}^{(2)}}, \\
        d_2^{h}(\mathbf{h}^{(1)}_{n,m}, \mathbf{h}^{(2)}_{n,m}) &= \sqrt{\sum_{i,j} \left((hF)_{i,j}^{(1)} - (hF)_{i,j}^{(2)} \right)^2}.
    \end{align*}
where $h\mathbf{F}^{(1)}$ and $h\mathbf{F}^{(2)}$ are the $h\mathbf{F}$-matrix embedding matrices associated with $\mathbf{h}^{(1)}_{n,m}$ and $\mathbf{h}^{(2)}_{n,m}$, respectively.

\subsection{Metrics on rooted unlabeled ranked phylogenetic networks with different number of hybridization events}
\label{sec:diff-hyb}

%
Consider two phylogenetic networks with $n$ leaves, $P^{(1)}_{n,m_{1}}$ and  $P^{(2)}_{n,m_{2}}$ with $m_{1}$ and $m_{2}$ hybridization events respectively. In order to use our distances, we require the two $\mathbf{F}$-matrix representations of $P^{(1)}_{n,m_{1}}$ and  $P^{(2)}_{n,m_{2}}$ to be of the same size. To do this, we insert artificial events as follows. For $i=1,2$, let $\mathbf{E}^{(i)}=(e^{(i)}_{n+2m_{i}},\ldots,e^{(i)}_{1})$ be the vector of ordered event types (hybridization and branching) events of the $i$-th network, with $e^{(i)}_{j} \in \{b,h\}$ at time $u^{(i)}_{j}$, where $b$ denotes a branching event and $h$ denotes a hybridization event. In the example illustrated in Figure \ref{fig:different_hybridization_network_examples}, the event vectors for $P^{(1)}_{n,m_{1}}$ and  $P^{(2)}_{n,m_{2}}$ are $\mathbf{E}^{(1)}=(b,b,h,b,b,b)$ and $\mathbf{E}^{(2)}=(b,h,b,h,b,b,b,b)$.

We first align all branching events between the two networks by adding empty spaces from left to right as depicted in Figure \ref{fig:align}. Once all the type-b events are aligned, we next align the hybridization events between two consecutive branching events or before the first branching event, from left to right. The event vector alignment is demonstrated in Figure \ref{fig:align}.  If one network has more type h events than the other in a given interval between branching events, we insert the excess artificial events, denoted by a's in the network with fewer type h events in that interval. We then augment the two weight matrices $\mathbf W^{(1)}$ and $\mathbf W^{(2)}$ to match the dimensions of the augmented $\mathbf{F}$ matrices. If $n_{a}$ artificial events are inserted between events $e^{(i)}_{j+1}$ and $e^{(i)}_{j}$, we subdivide the corresponding time interval $[u^{(i)}_{j+1},u^{(i)}_{j}]$ into $n_{\alpha}+1$ intervals with equal length $\Delta=\frac{u^{(i)}_{j}-u^{(i)}_{j+1}}{n_{a}+1}$. Figure \ref{fig:aligned_networks} shows the two aligned networks.
%



\begin{figure}[H]
\begin{subfigure}[b]{0.45\textwidth}
\centering
\includegraphics[width = \textwidth]{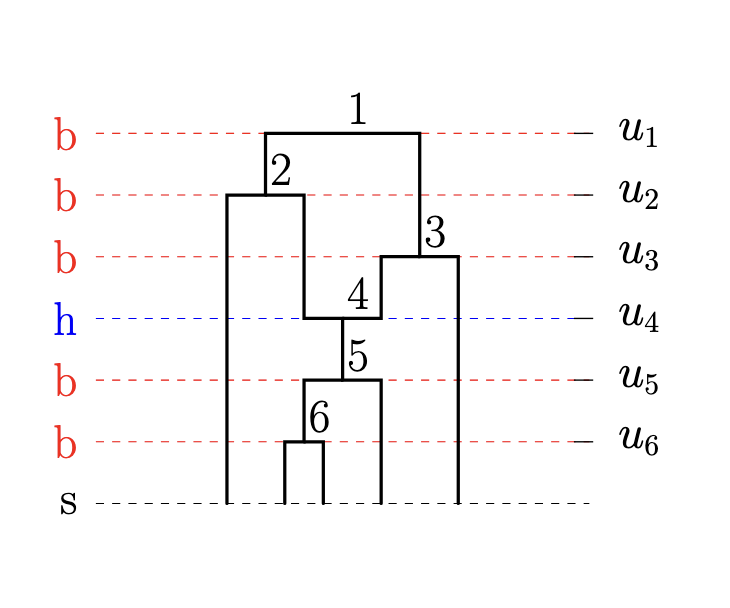}
\caption{}
\label{fig:different_hybridization_network_example1}
\end{subfigure}
\hspace{0.05\textwidth}
\begin{subfigure}[b]{0.45\textwidth}
\centering
\includegraphics[width = \textwidth]{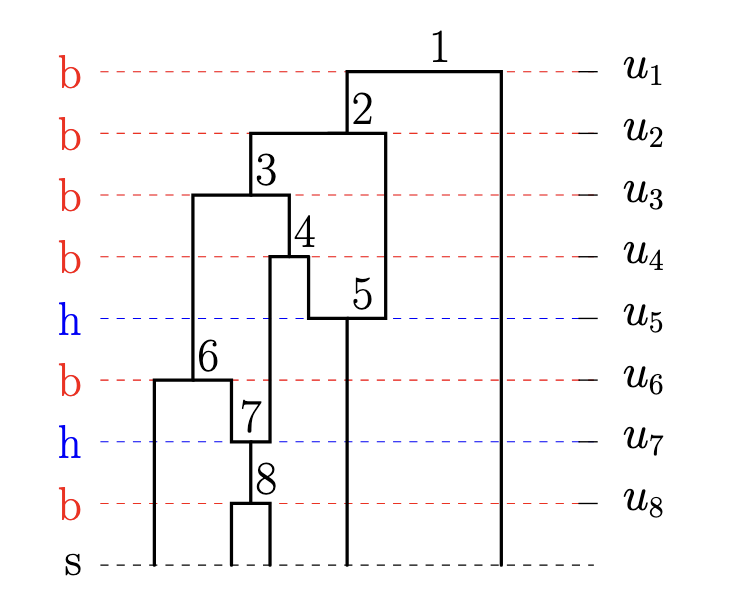}
\caption{}
\label{fig:different_hybridization_network_example2}
\end{subfigure}
\caption{(A) $P^{(1)}_{n,m_{1}}$ and (B) $P^{(2)}_{n,m_{2}}$ are two phylogenetic networks with the same number of leaves but different number of hybridization events. Specifically, they both have five leaves, but $P^{(1)}_{n,m_{1}}$ has one hybridization while $P^{(2)}_{n,m_{2}}$ has two. Each event is labeled as $b$, a branching event, or $h$, a hybridization event.}
\label{fig:different_hybridization_network_examples}
\end{figure}


\begin{figure}[H]
\includegraphics[width=0.8\textwidth]{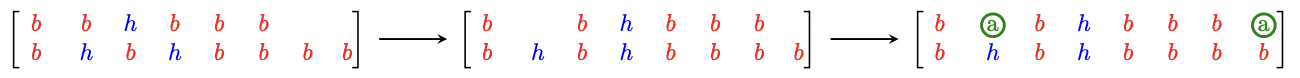}
\centering
\caption{Alignment of event vectors of two networks. Branching events (b) are aligned first and then hybridization events (h) are aligned between two branching events or before the first branching event. In this example, since h events are automatically aligned after the b events are aligned, we omitted the additional step of aligning h events. Then we insert the excess artificial events a's in the empty spaces.}
\label{fig:align}
\end{figure}

\begin{figure}[H]
\begin{subfigure}[b]{0.45\textwidth}
\centering
\includegraphics[width = \textwidth]{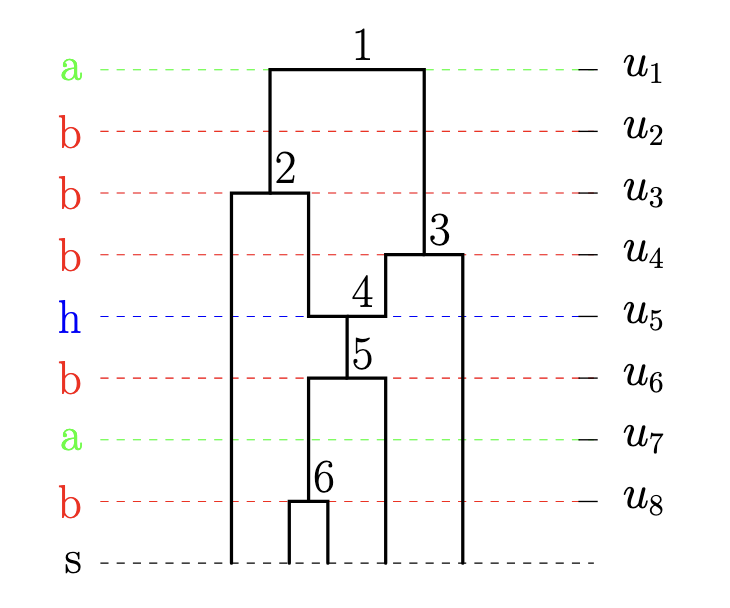}
\caption{}
\label{fig:aligned_network1}
\end{subfigure}
\hspace{0.05\textwidth}
\begin{subfigure}[b]{0.45\textwidth}
\centering
\includegraphics[width = \textwidth]{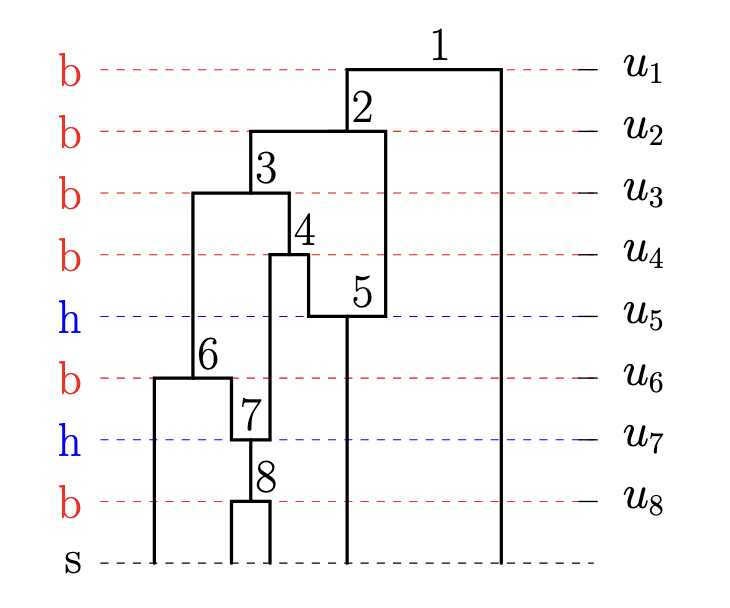}
\caption{}
\label{fig:aligned_network2}
\end{subfigure}
\caption{Aligned ranked phylogenetic networks (original networks in Figure \ref{fig:different_hybridization_network_examples}). Each event is labeled as $b$, a branching event, $h$, a hybridization event, or $a$, artificial event added for alignment.}
\label{fig:aligned_networks}
\end{figure}



\section{Probability models on networks}

We will test our proposed metrics in Section~\ref{sec:simulations} on simulations generated according to a recently proposed beta-splitting model on rooted unlabeled ranked phylogenetic networks \citep{zhong2026beta}. A wide range of probabilistic models have been developed to describe the diversification of lineages on tree topologies, including the Yule (pure-birth) \citep{Yule1925}, and coalescent models \citep{wakeley_coalescent_2008}, that generate uniform tree topologies, and the one parameter beta-splitting model \citep{Aldous1996,Blum2006} that generates trees with different balance/unbalance tendencies depending on the value of the parameter. In contrast, analogous models for phylogenetic networks are still lacking. The  forward in-time birth–hybridization model \citep{Zhang2018}, and the backward in-time ancestral recombination graph \citep{griffiths1997ancestral}, both assume the phylogenetic network topology is uniform. Recently, \citet{janssen2021comparing} introduced a beta-splitting model on networks that first generates a tree topology according to a beta-splitting model, and then hybrid edges are randomly added on the tree. In contrast, \cite{zhong2026beta} present a one-parameter ($\beta\geq -1$) beta-splitting model that generates a network topology directly given the sequence of speciation and hybridization events. We state their algorithm  (Alg.~\ref{alg:beta_splitting}) for completeness. Similar to the beta-splitting model for trees, this model generates ``unbalanced" or ``ladder" like network topologies for small values of $\beta$ and more ``balanced" with large number of cherries for large values of $\beta$. A cherry is a subtree subtending two leaves. 

To generate the sequence of speciation and hybridization events, we used \citet{zhong2026beta} birth-hybridization model conditioned on a fixed number of leaves and hybridization events. The birth-hybridization continuous-time Markov Chain is denoted by $(R(t),Y(t))$, that keeps track of the number $R(t)$ of hybridizations  at time $t$ and the total number $Y(t)$ of lineages at time $t$, conditioned on experiencing $m$ reticulations and reaching $n$ lineages only once after $m$ hybridizations have occurred. The initial state is $(R(0),Y(0))=(0,1)$ and the transition rates are:

\begin{align}
  \lambda_{(r,i),(s,j)}=\begin{cases} i \mu & j=i+1, s=r \text{ birth/speciation}\\
    \binom{i}{2} \rho & j=i-1, s=r+1, i>1, r<m \text{ death/reticulation}\\
    -i\mu-\binom{i}{2} \rho \mathbbm{1}(r<m) & j=i, i>1\\
    -\mu & j=i=1\\
    0 & \text{otherwise.}
    \end{cases} 
\end{align}

In particular, to avoid identifiability issues (see Lowest Stable Ancestor (LSA) in \citep{hector2017network}), we assume the first two transitions are $(0,1)$ to $(0,2)$, and $(0,2)$ to $(0,3)$ with rates $\mu$ and $2\mu$ respectively. Given a full realization of the Markov chain, we can extract the sequence $s_{1},\ldots,s_{n+2m-1}$ that indicate whether the $i$-th event was a speciation $s_{i}=1$ or a hybridization $s_{i}=-1$. This sequence is then the input of Algorithm~\ref{alg:beta_splitting}.



\begin{algorithm*} 
 \caption{Simulation algorithm for the beta-splitting model}\label{alg:beta_splitting}
    \begin{algorithmic}
        \State \hspace{0.4em} \textbf{Input:} 
        $s_{1},\ldots,s_{n+2m-1}$, where $s_{i}=1$ indicates speciation and $s_{i}=-1$ indicates hybridization
        \State \hspace{4.4em} $\beta$ (beta-splitting model parameter, $\beta\geq 1$)

        \State \hspace{0.4em} \textbf{Output:} A rooted ranked unlabeled network
        \State \hrulefill
    \end{algorithmic}

    \begin{algorithmic}[1]
        \State Construct the generative sequences: $B_{1}, B_{2}, ..., B_{n+2r-2} \sim \text{Beta} (\beta+2, \beta+2)$; $U_{1}, U_{2}, ..., U_{n+2r-2} \sim \text{Unif} [0,1]$; $V_{1,1}, V_{1,2}, V_{2,1}, V_{2,2}, ..., V_{r,1}, V_{r,2} \sim \text{Unif} [0,1]$
        \State The root is labeled to be the interval $[0,1]$ 
        \State At step 1, we split the root node into a left leaf node labeled by $[0,b_{1}]$ and a right leaf node labeled by $[b_{1},1]$, where $b_{1}$ is the value of $B_{1}$. Change the root label to be the integer 1, indicating that the root speciates into two lineages at time point 1.
        \For{step $k \in \{2,3,...,n+2r-2\}$}
            \If{$s_{k}=1$ event}
                \State Find the leaf node whose interval label $\cup_{t \in T} [x_{t},y_{t}]$ contains $U_k$. 
                \State Change this leaf's label to the integer $k$, indicating that the node speciates into two lineages at time point $k$
                \State Split the leaf into a left leaf node with label $\cup_{t \in T} [x_{t},x_{t} + (y_{t} - x_{t}) * b_k]$ and a right leaf node with label $\cup_{t \in T} [x_{t} + (y_{t} - x_{t}) * b_k, y_{t}]$
            \Else
                \State Find the two leaf nodes whose interval labels $\cup_{t \in T_{1}} [x_{t},y_{t}]$ and $\cup_{t \in T_{2}} [x_{t},y_{t}]$ contain $V_{k,1}$ and $V_{k,2}$, respectively
                \State Change both their labels to be the integer $k$, indicating that the two lineages hybridize at time point $k$
                \State Merge these two leaves into a single leaf node with label $\cup_{t \in T_{1} \cup T_{2}} [x_{t},y_{t}]$.
            \EndIf
        \EndFor
    \end{algorithmic}
\end{algorithm*}

\section{Results}

In this section, we evaluate the ability of our metrics to distinguish between different sampling distributions and show the applicability of our metrics in analyzing posterior distributions of viral networks.

To our knowledge, there are no other distances defined on rooted, ranked and unlabeled networks. For this reason, we compare the performance of our distances to distances defined on labeled networks. To have a valid comparison, we labeled each unlabeled phylogenetic network in a unique way as follows. 
We impose an induced labeling on the unlabeled network by ordering its leaf nodes according to the recency of their most recent common ancestors. Leaves descending from more recent ancestral events are assigned smaller label indices.

We will compare our distance with the following distances:

\textbf{Hamming distance:} This distance is computed between the corresponding adjacency matrices.  The two matrices are then compared entry by entry, and the distance is defined as the total number of entries for which the two matrices differ. In this way, the Hamming distance counts how many edges must be added or removed to transform one network into the other.

\textbf{Hardwired cluster distance \citep{Huson2010}:}  The distance is computed by extracting the collection of clusters induced by its edges from each phylogenetic network, where each cluster corresponds to the set of descendant leaves below that edge. The two resulting cluster sets are then compared, and the distance is defined as the number of clusters that appear in one network but not in the other.

\textbf{Diffusion distance \citep{Hammond2013}:}  The distance is computed by comparing diffusion processes induced by the graph Laplacians. For each network, a heat kernel is obtained by exponentiating its Laplacian, which characterizes how information diffuses across the network over time. The diffusion distance between two networks is then defined as the Frobenius norm of the difference between their heat kernels, quantifying differences in how the two network structures transmit information.


\subsection{Simulated data} \label{sec:simulations}

We generated 100 rooted unlabeled and ranked phylogenetic networks with $n=100$ tips and $m=10$ hybridizations from five different beta-splitting distributions on network topology with parameters:  $\beta = \{-0.9, -0.5, 0.0, 1.0, 100\}$. To simulate the sequence of speciations and hybridizations, we assumed  $\mu=4.0$ and $\rho=0.4$. Varying the parameter $\beta$ produced markedly different phylogenetic network topologies. The first row of Figure~\ref{fig:simulated_five_beta_values} shows the medoid networks from each $\beta$ group, illustrating systematic changes in topological balance: networks generated with $\beta = -0.9$ are the most unbalanced, whereas those with $\beta = 100$ are the most balanced. 
The second row of Figure \ref{fig:simulated_five_beta_values} shows how the histograms of cherry counts vary across the range of $\beta$ values. As expected, networks simulated under distributions with small value of $\beta$ have smaller number of cherries than networks simulated with large value of $\beta$. In a phylogenetic network, a cherry is defined as a pair of leaves that share an immediate common ancestor. Cherry counts provide a simple yet informative summary of network topology, as networks with differing shapes often exhibit distinct numbers of such leaf pairs. The observed variation in cherry counts across $\beta$ values thus offers additional evidence for topological differences among the network groups.
 
We computed pairwise distance matrices among all simulated networks for each of the four distances under consideration (our $d_{2}$ metric, hardwired cluster, Hamming, and diffusion). Multidimensional scaling (MDS) was then applied to visualize the relative positions of the ranked phylogenetic network shapes across the different $\beta$ values.
Figure \ref{fig:four_dist} compares the distribution discrimination performance of the four metrics and demonstrates that our $d_2$ metric more effectively discriminates among the network topologies than the alternative approaches. Moreover, our $d_2$ metric exhibits the shortest computational running time (Table \ref{tab:running_time}).


\begin{figure}[H]
\centering
\hspace{0.05cm}\includegraphics[scale=0.10]{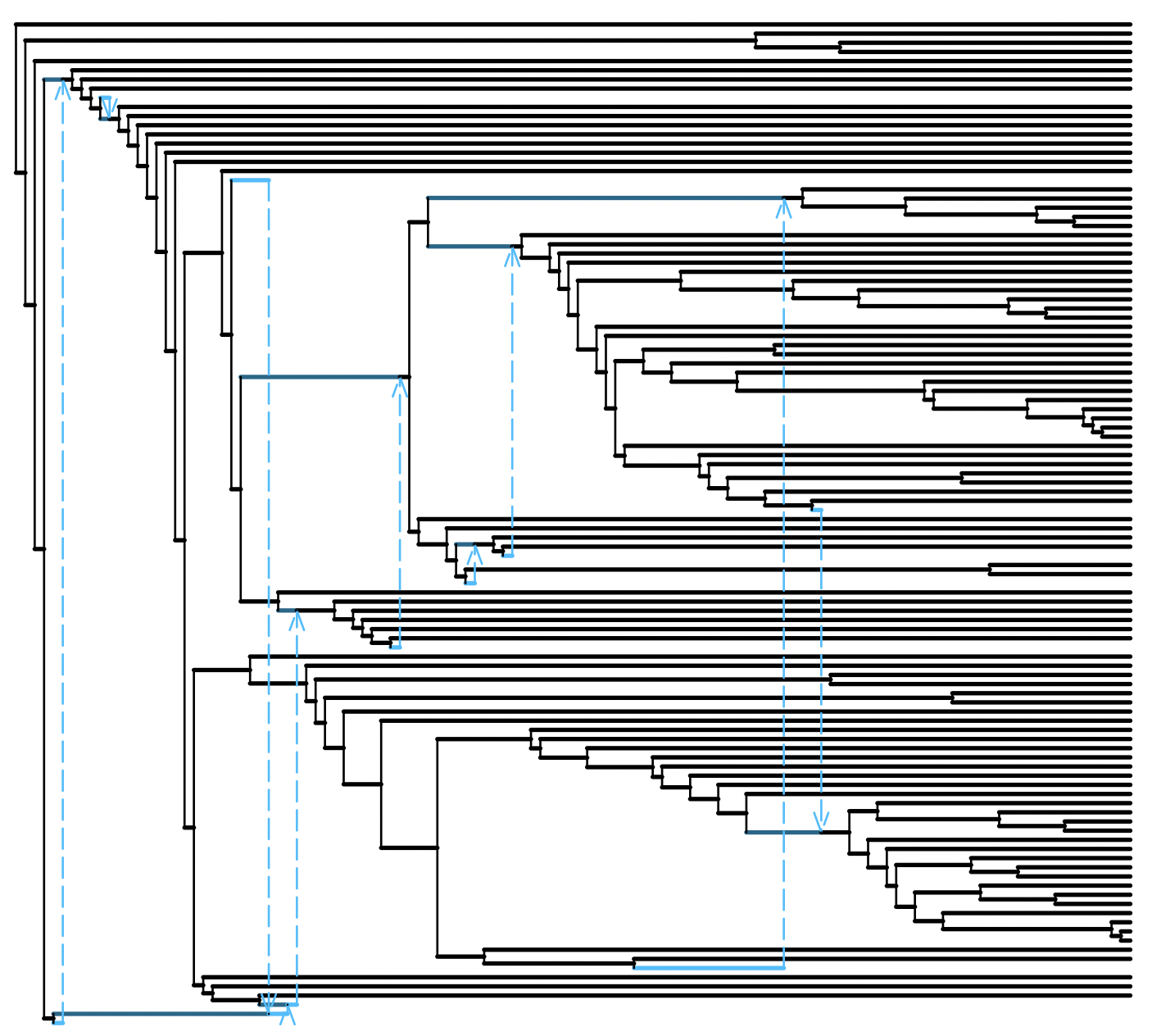}
\hspace{0.6cm}\includegraphics[scale=0.10]{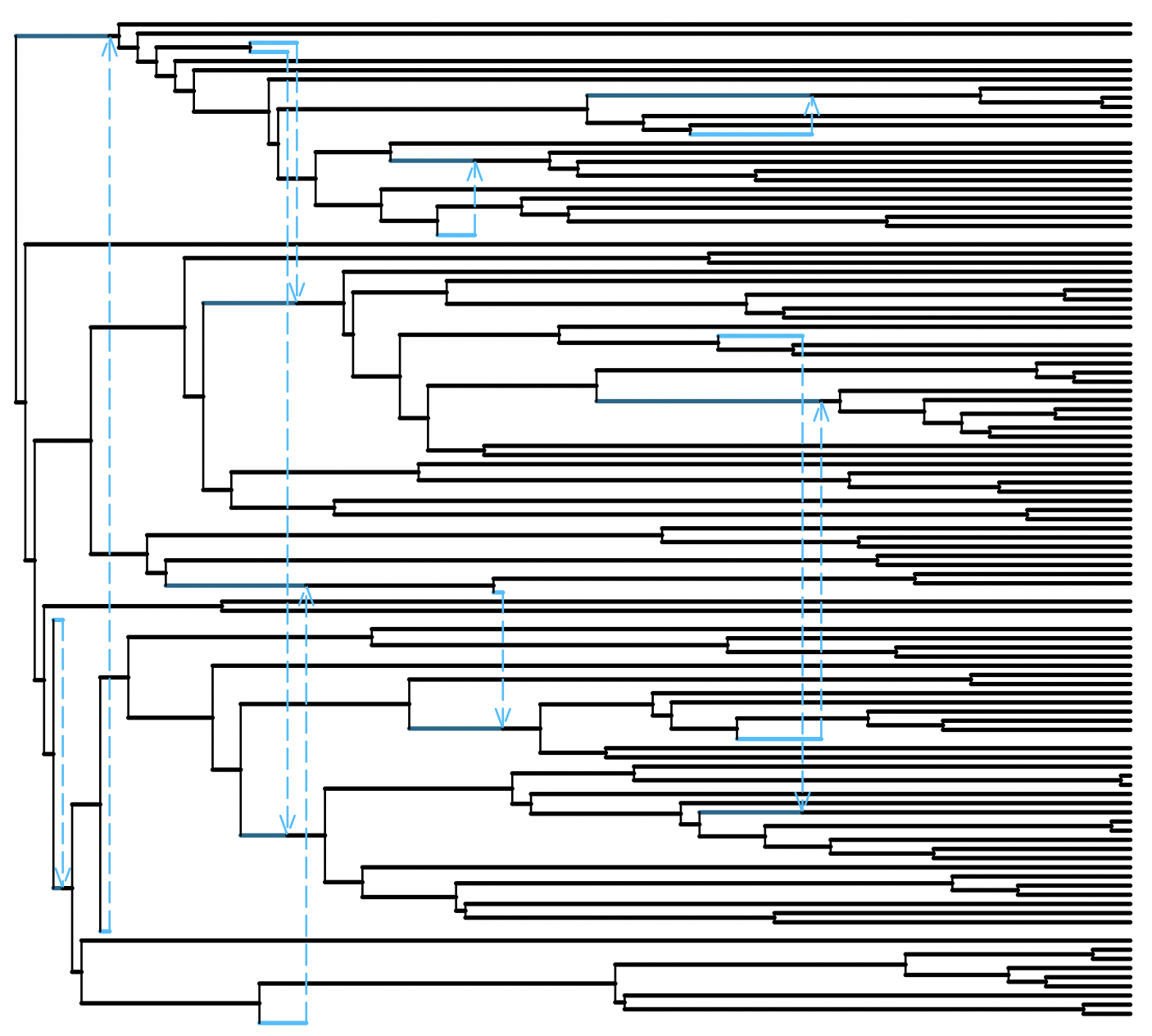}
\hspace{0.6cm}\includegraphics[scale=0.10]{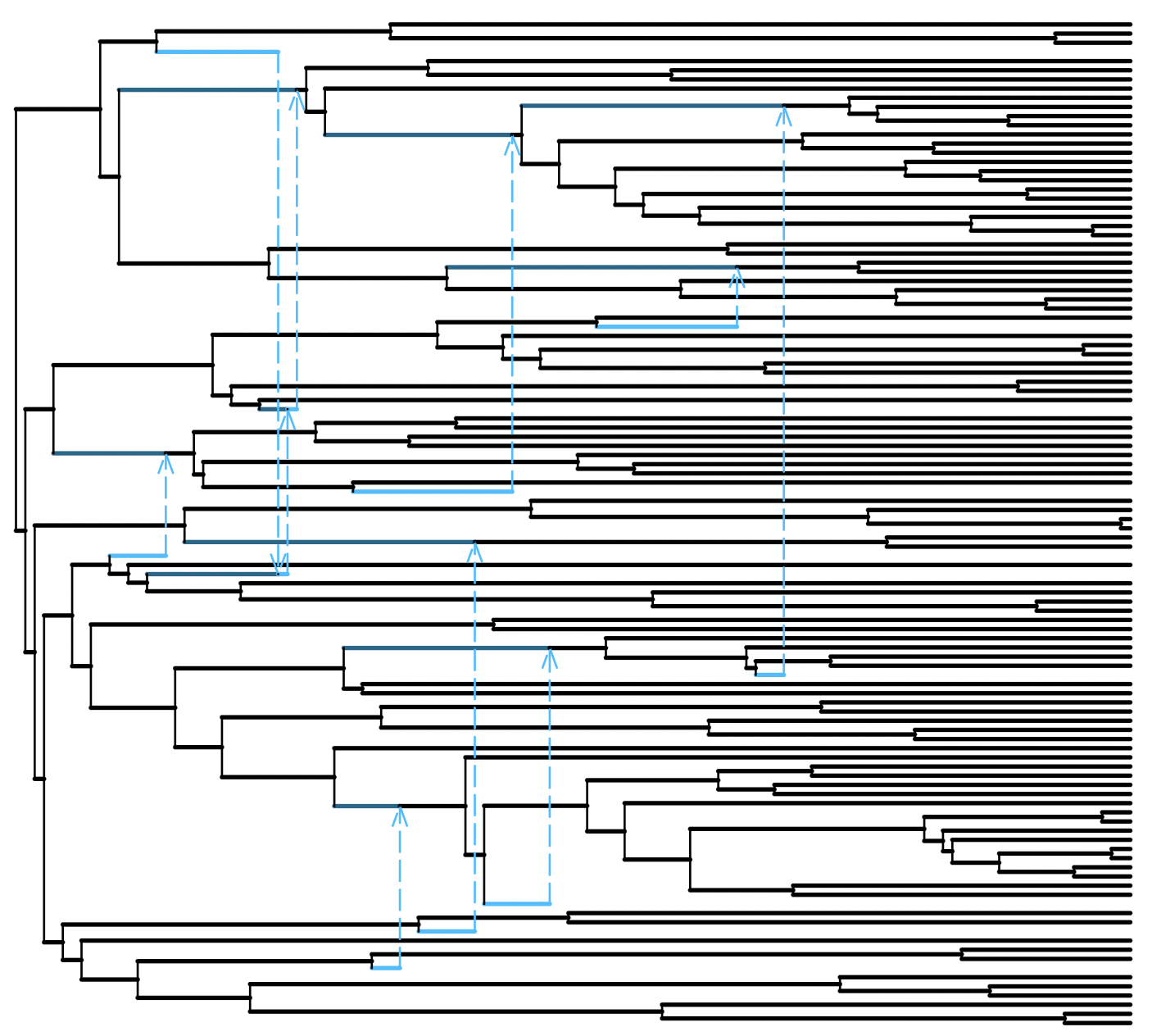}
\hspace{0.6cm}\includegraphics[scale=0.10]{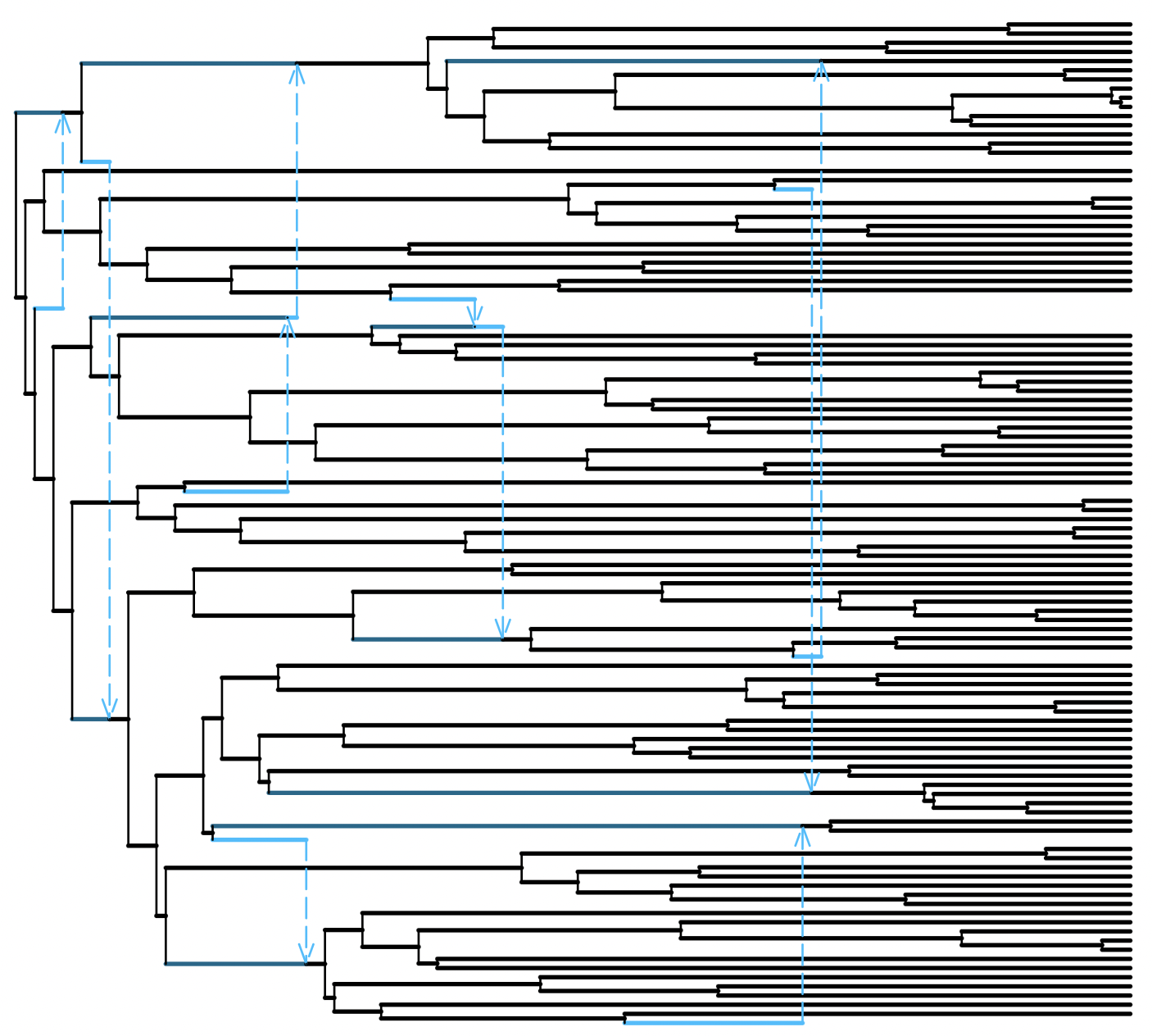}
\hspace{0.6cm}\includegraphics[scale=0.10]{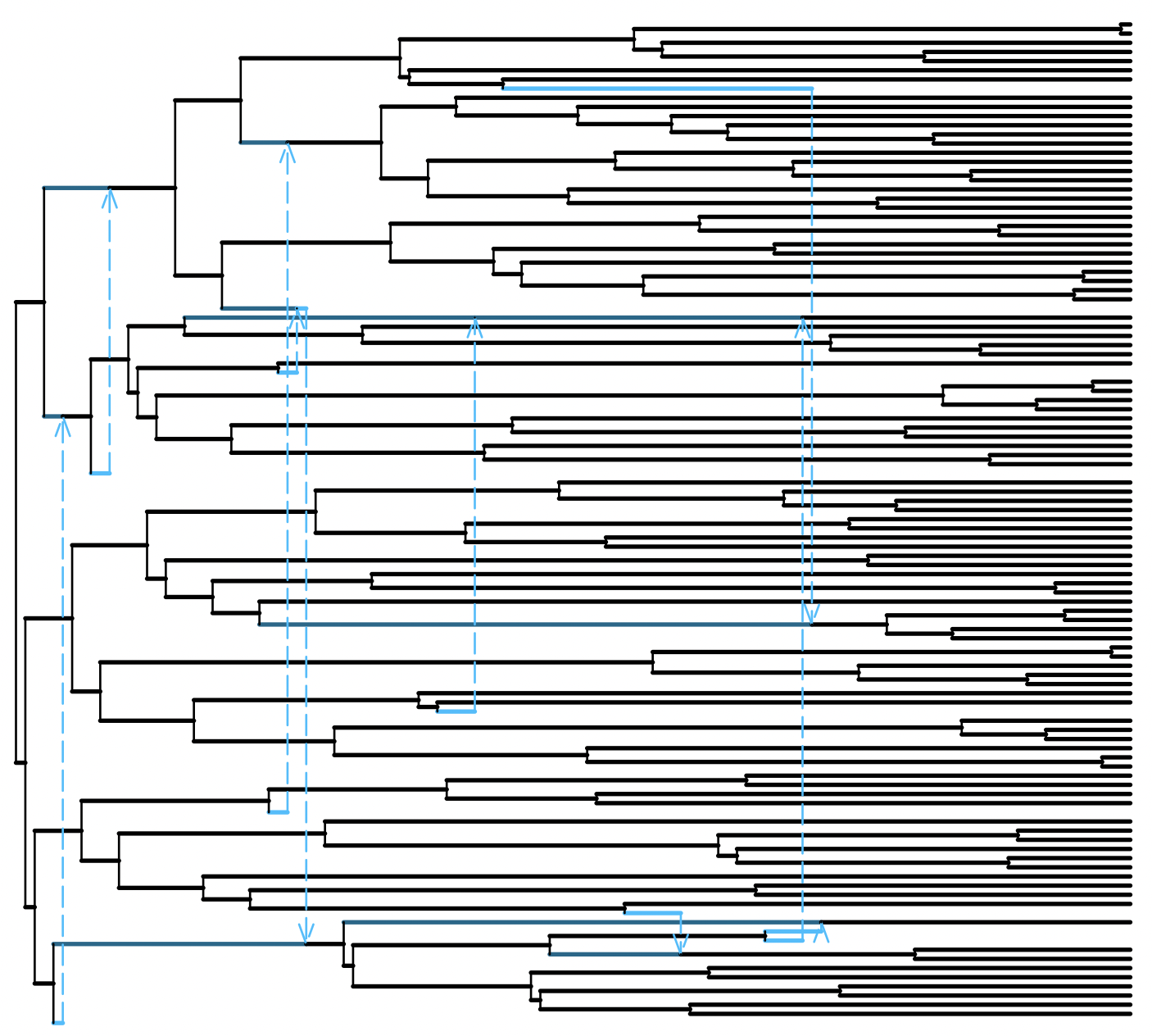}\\[0.5cm]
\includegraphics[scale=0.15]{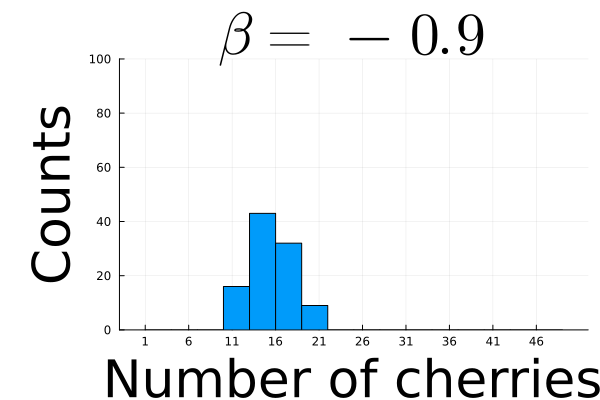}
\includegraphics[scale=0.15]{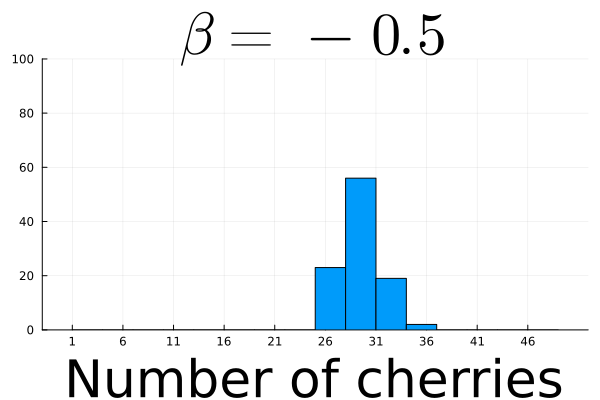}
\includegraphics[scale=0.15]{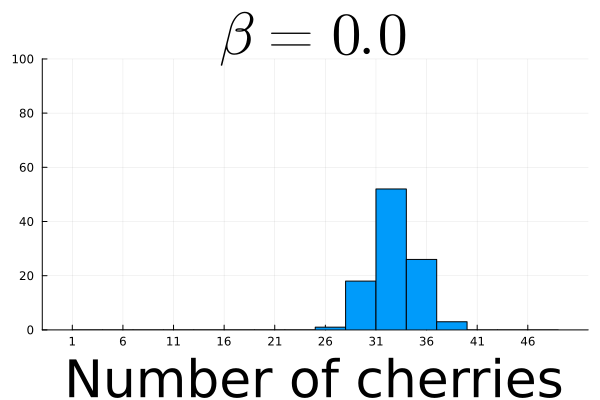}
\includegraphics[scale=0.15]{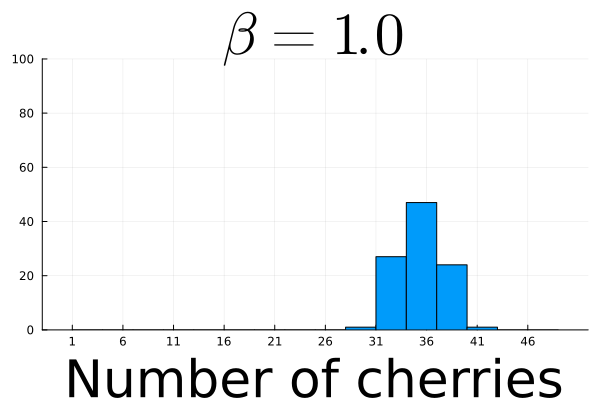}
\includegraphics[scale=0.15]{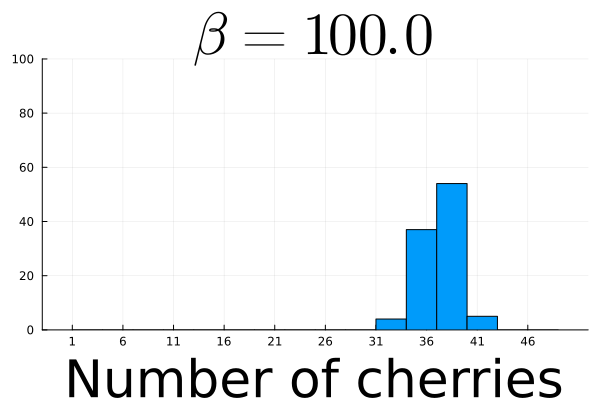}
\caption{Top: Medoid networks with 100 tips corresponding to different $\beta$ values: $\beta = -0.9$, $\beta = -0.5$, $\beta = 0.0$, $\beta = 1.0$, $\beta = 100.0$. Bottom: The number of cherries varies significantly across different $\beta$ values. Specifically, there is a notable difference between the groups with $\beta = -0.9$ and $\beta = -0.5$ compared to the other three groups. Although the networks generated with $\beta = 0.0$, $\beta = 1.0$ and $\beta = 100.0$ exhibit similar average numbers of cherries, their distributions of the number of cherries still somewhat differ, indicating underlying topological variations among these groups.}
\label{fig:simulated_five_beta_values}
\end{figure}

\begin{table}[h]
\centering
\begin{tabular}{|c|c|c|c|c|}
\hline
               & $d_2$ (our distance)  & Hardwired Cluster & Hamming & Diffusion \\ \hline
Time (seconds) & 9.12   & 670.80  &  59.14  &  137.85   \\ \hline
\end{tabular}
\vspace{0.5cm}
\caption{Running times for computing distance matrices in the simulation using four different distance methods. The running times were measured on a local laptop with Apple M1 Pro processor, 16 GB RAM, using Julia v1.10.}
\label{tab:running_time}
\end{table}

\begin{figure}[H]
\begin{subfigure}[b]{0.45\textwidth}
\centering
\includegraphics[width = \textwidth]{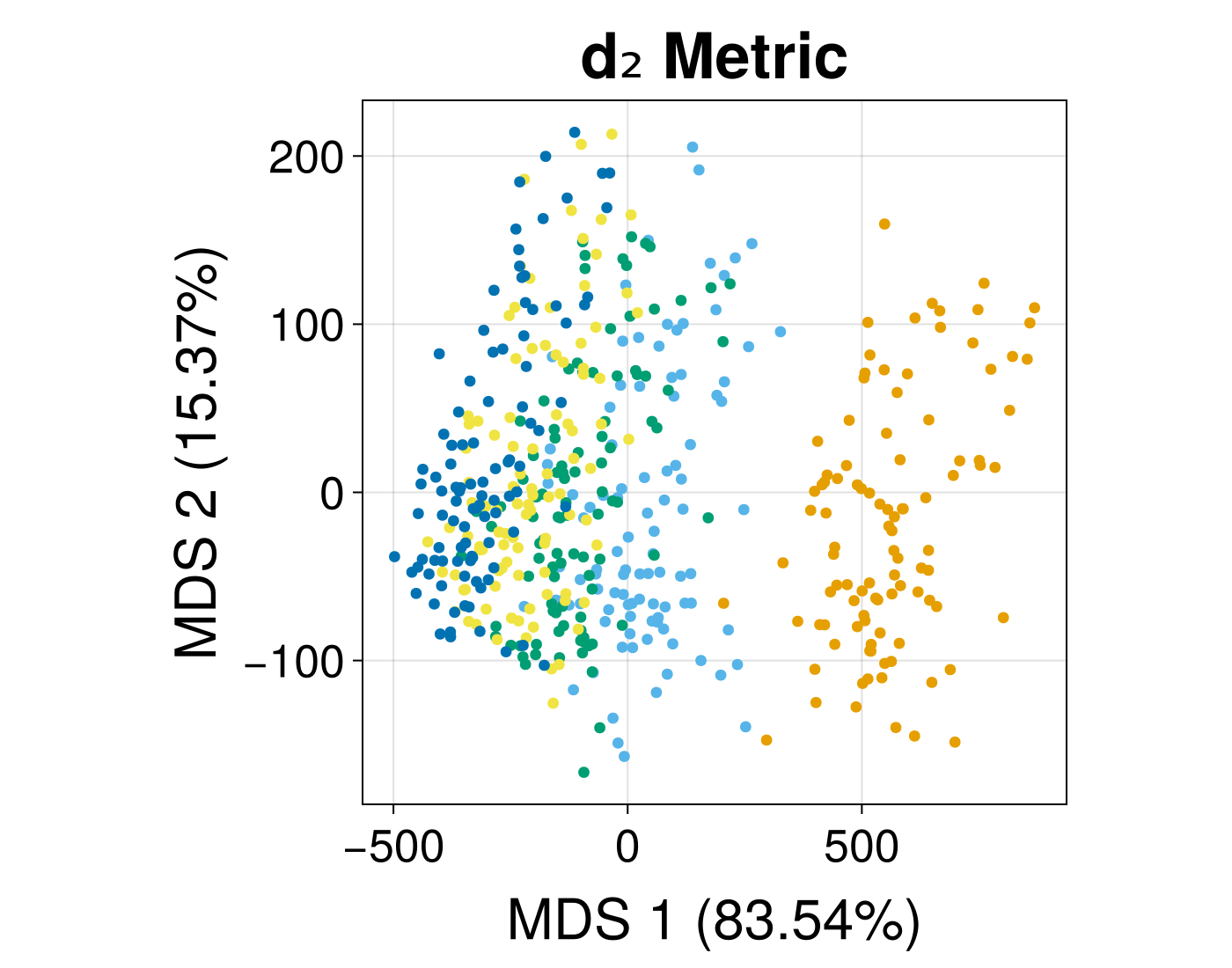}
\caption{}
\label{fig:Our_Metric}
\end{subfigure}
\begin{subfigure}[b]{0.45\textwidth}
\centering
\includegraphics[width = \textwidth]{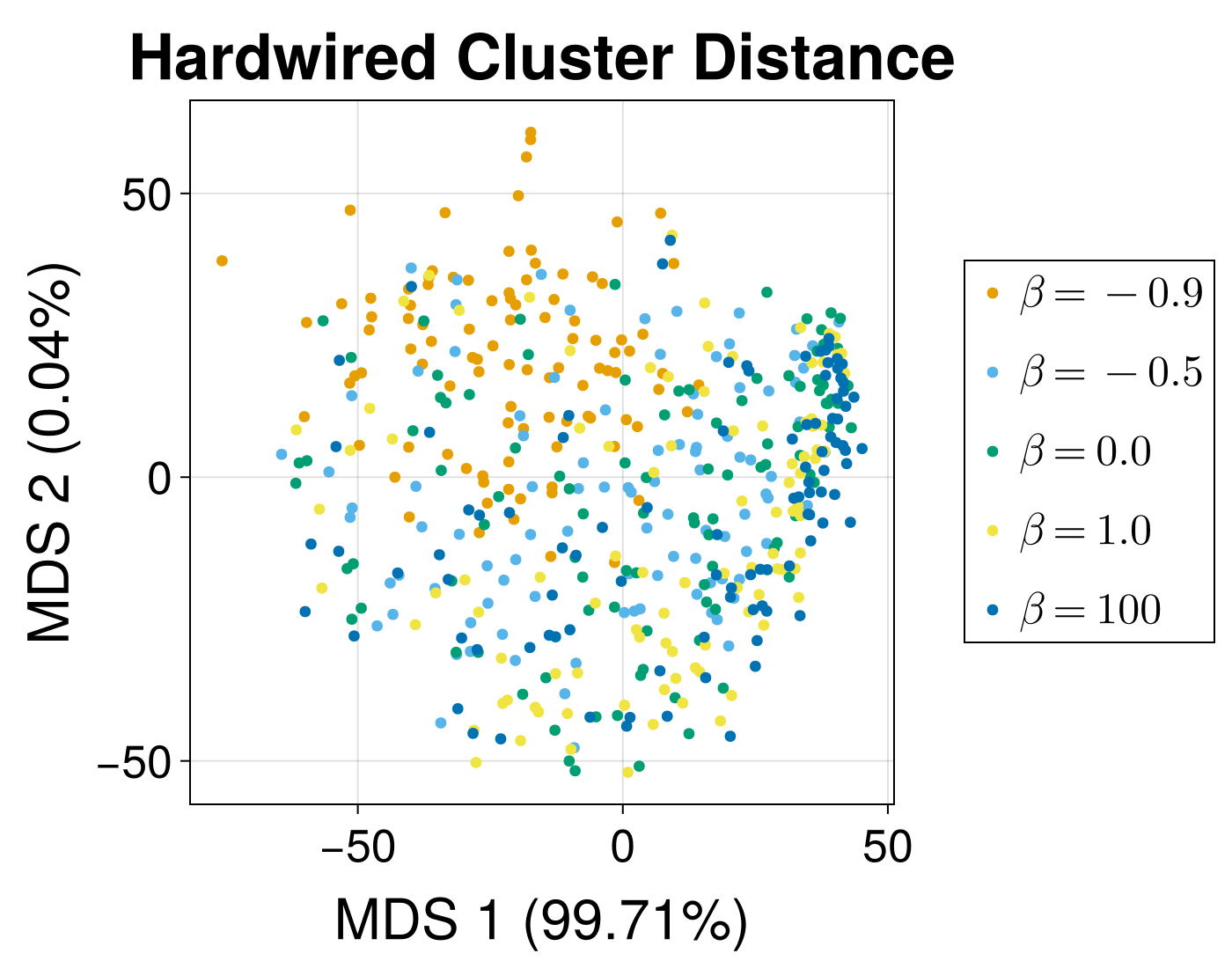}
\caption{}
\label{fig:Hardwired_Distance}
\end{subfigure}
\begin{subfigure}[b]{0.45\textwidth}
\centering
\includegraphics[width = \textwidth]{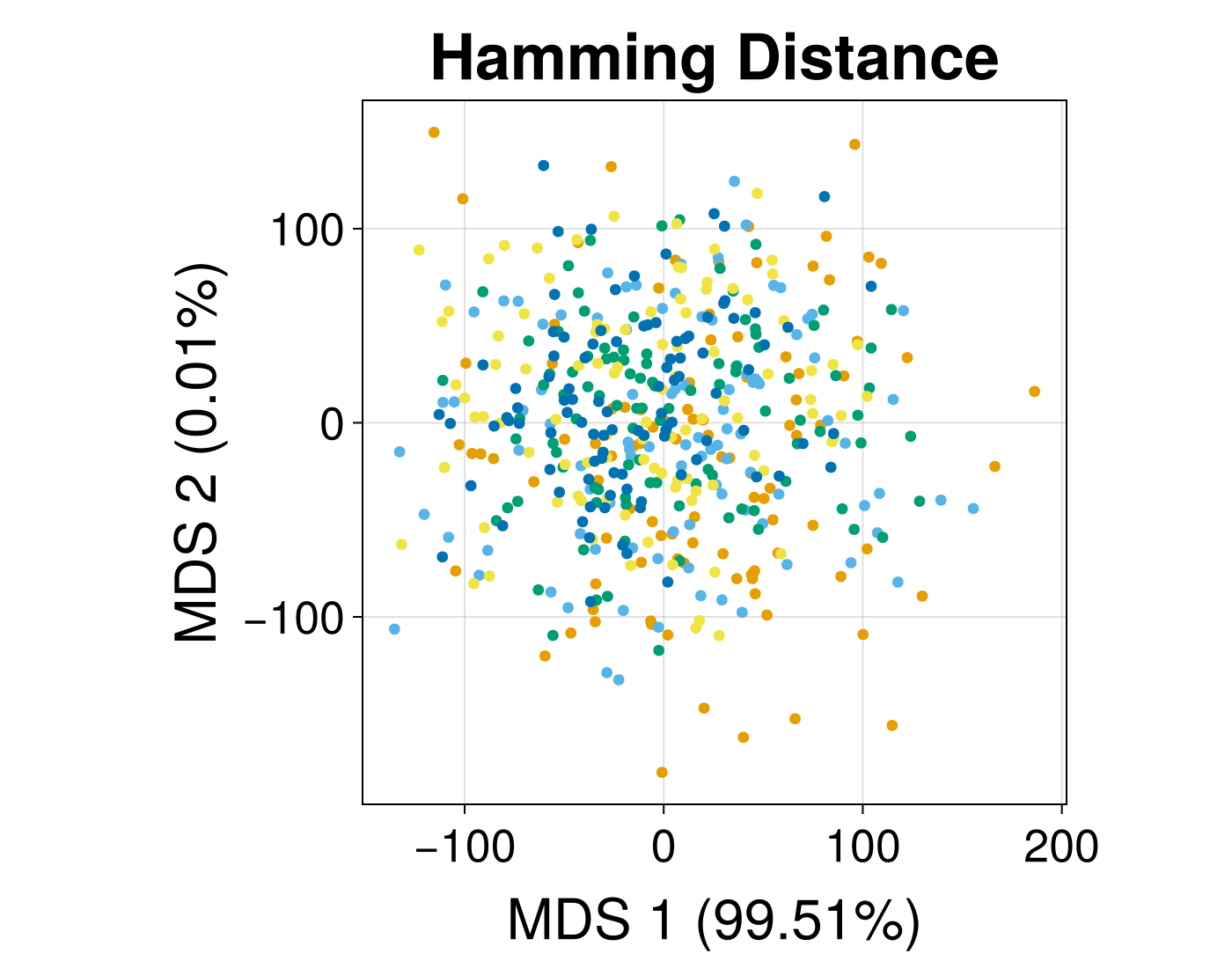}
\caption{}
\label{fig:Hamming_Distance}
\end{subfigure}
\hspace{0.09\textwidth}
\begin{subfigure}[b]{0.45\textwidth}
\centering
\includegraphics[width = \textwidth]{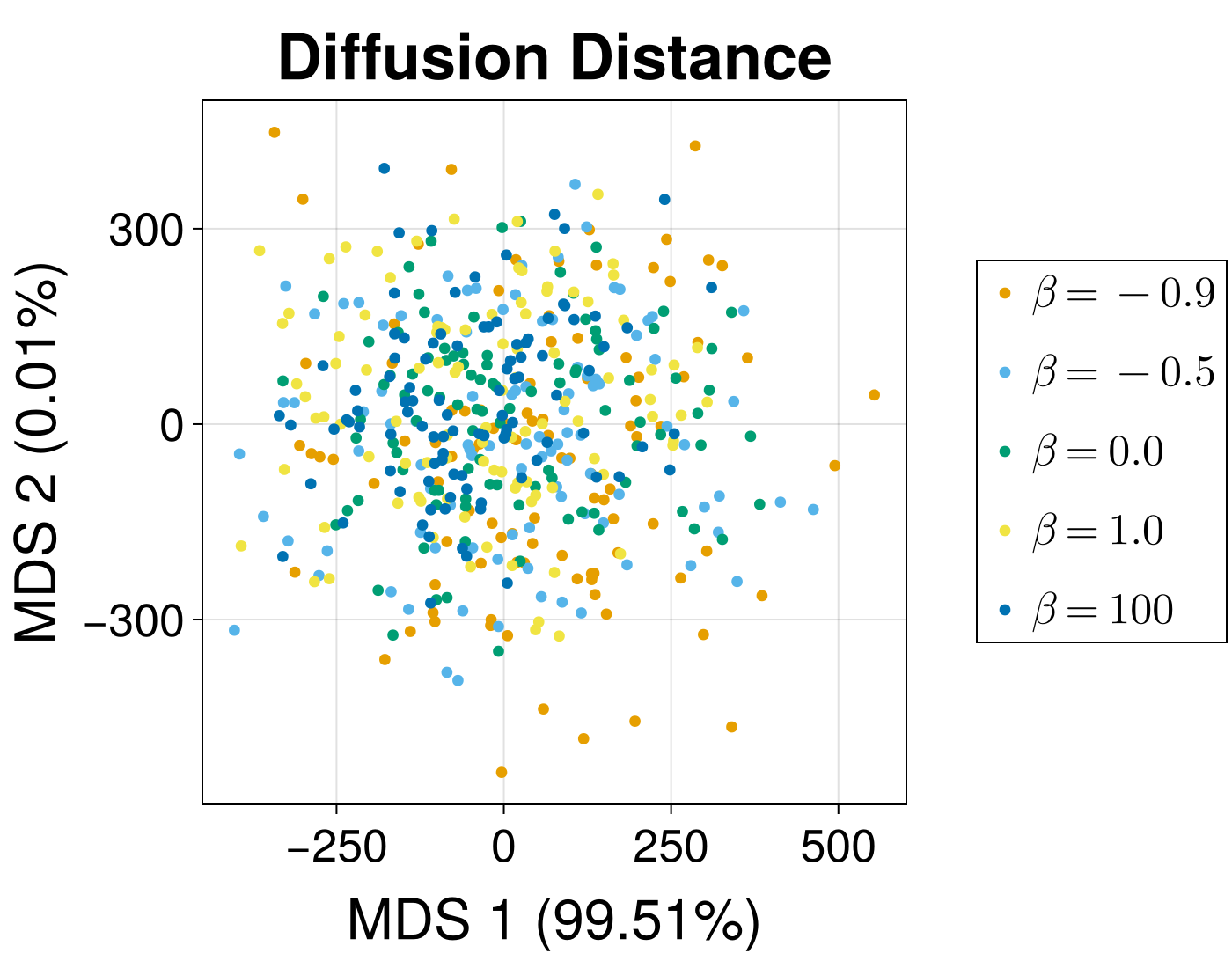}
\caption{}
\label{fig:Diffusion_Distance}
\end{subfigure}
\caption{Multidimensional scaling (MDS) representation of distances between 500 simulated isochronous ranked network shapes with n=100 leaf nodes, generated under different beta-splitting models. A total of 110 samples for each group are generated, of which the initial 10 were discarded as burn-in. For each of $\beta$ values in $\{ -0.9, -0.5, 0.0, 1.0, 100.0\}$, 100 isochronous ranked network shapes were simulated. The MDS plots illustrate the clustering patterns based on different distance metrics: (A) $d_2$ metric, (B) Hardwired Cluster Distance, (C) Hamming Distance, (D) Diffusion Distance.}
\label{fig:four_dist}
\end{figure}


\subsection{Comparing the evolutionary histories of influenza A/H1N1 across different regions}

We compared the evolutionary histories of A/H1N1 influenza viruses across three geographic regions: the United States, Europe and Asia. Viral molecular sequences were obtained from \href{https://doi.org/10.55876/gis8.260429ca}{GISAID} with identifiers available in \url{https://doi.org/10.55876/gis8.260429ca}. For each geographic region, we arbitrarily selected 40 human samples of A/H1N1 influenza viruses sampled on December 17 of 2024 in the United States, Europe and Asia, and obtained gene sequences for each of the five segments (PB1, HA, NP, NA and NS) from each sample, for a total of 200 sequences.
All segment sequences were aligned separately across all samples/patients using MAFFT \citep{Katoh2013}.
Phylogenetic network distributions for each gene and region were inferred using BEAST 2 \citep{Bouckaert2014} with the CoalRe package \citep{Muller2020}. From each posterior distribution, we randomly sampled 100 phylogenetic networks. For all analyses, the molecular clock rate was fixed at 0.0025 substitutions per site per year, and the reassortment rate prior was specified as an exponential distribution with mean 0.3, corresponding to an average of three reassortment events per lineage every 10 years, 
and in agreement with \citet{Muller2020}. Importantly, because our metric operates on unlabeled phylogenetic networks, it permits direct comparison of network topologies derived from different samples/patients.


We computed pairwise distances among the sampled phylogenetic networks using our $d_2$ metric. Although all networks contained a fixed number of tips ($n=40$), they differed in their numbers of hybridization events. Accordingly, we employed the alignment strategy described in Section \ref{sec:diff-hyb} for comparing rooted, unlabeled, ranked phylogenetic networks with varying hybridization counts.
We then applied multidimensional scaling (MDS) (Figure \ref{fig:three_regions}) to visualize the resulting distances, where each point corresponds to a sampled network from the posterior distribution of a given region. To further characterize differences in network structure, we generated cherry-count histograms for each region. 
The MDS visualization showed no clear separation between United States and Europe, whose points overlapped extensively, whereas networks from Asia formed a distinct cluster. However, the MDS approach had a greater power of discrimination than the cherry distributions alone. The average number of cherries in the United States and Europe were 10.44 and 10.26, respectively, compared with a slightly higher average of 11.14 in Asia. In addition, the standard deviations for the United States and Europe were also similar, at 0.90 and 0.97, respectively.
%
Posterior mean estimates for key parameters, including the reassortment rate, population size, height, total tree length, and number of reassortment nodes are summarized in Table \ref{tab:posterior_mean_values}.

%

\begin{figure}[H]
\begin{subfigure}[t]{0.45\textwidth}
\centering
\includegraphics[height=6cm, width = \textwidth]{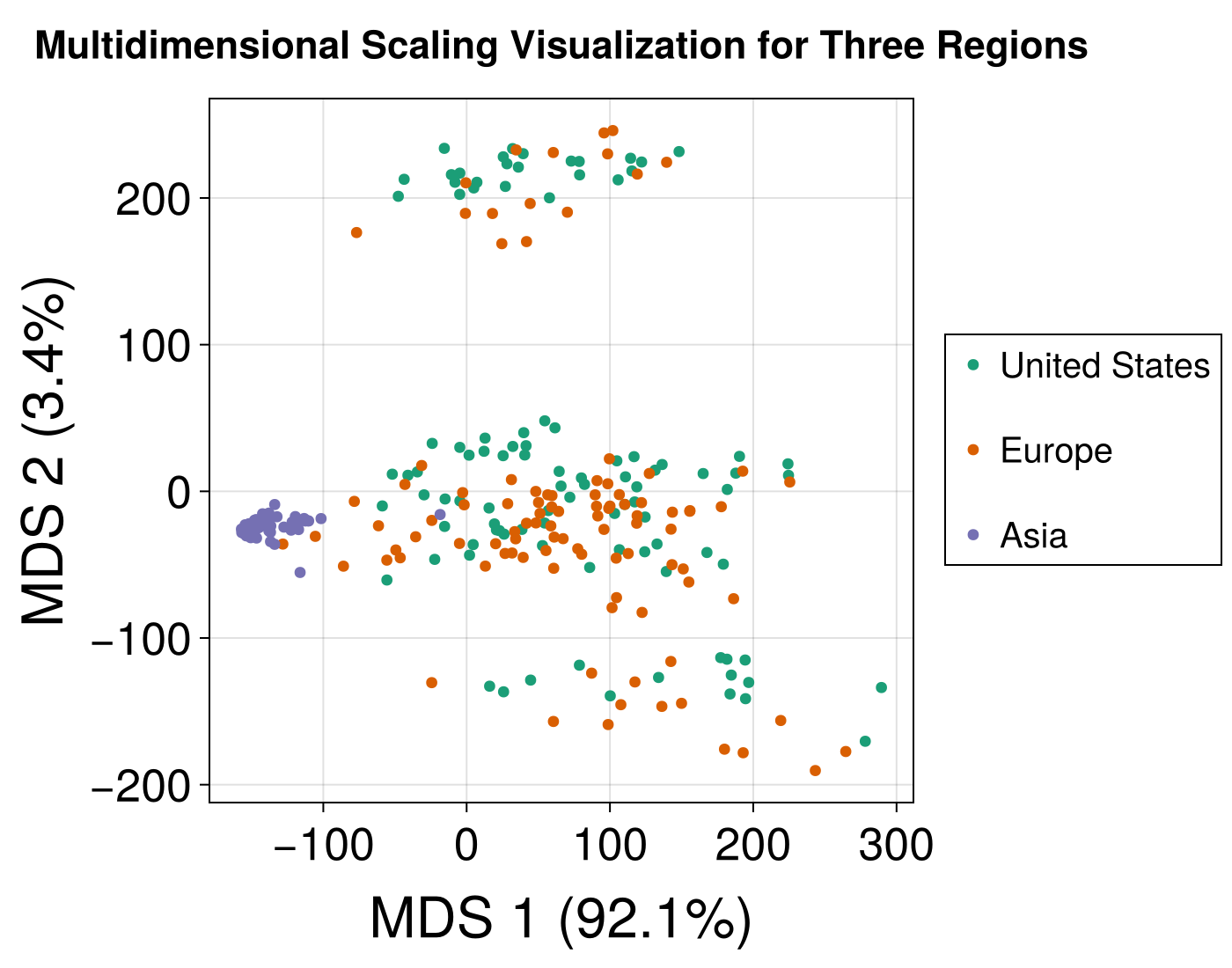}
\caption{}
\label{fig:mds_three_regions}
\end{subfigure}
\hspace{0.05\textwidth}
\begin{subfigure}[t]{0.45\textwidth}
\centering
\includegraphics[height=6cm, width = \textwidth]{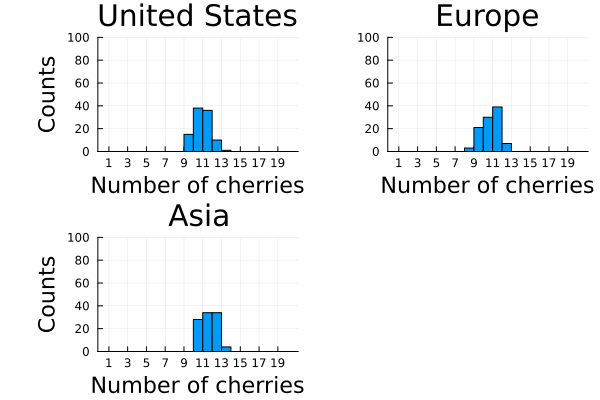}
\caption{}
\label{fig:cherry_numbers_for_three_regions}
\end{subfigure}
\caption{Comparison of the evolutionary histories of A/H1N1 influenza viruses across the United States, Europe and Asia. (A) Multidimensional scaling representation for the three different regions. (B) Histograms of cherry counts in phylogenetic networks by region. }
\label{fig:three_regions}
\end{figure}


\begin{table}[h]
\centering
\scalebox{0.8}{
\begin{tabular}{|c|c|c|c|c|c|}
\hline
       & Reassortment Rate & Population Size & Height & Total Length & Reassortment Node Count
       \\ \hline
United States  &  0.268 & 8.225 & 11.832 & 60.286  & 10.213       \\ \hline
Europe        & 0.353  & 5.524  & 10.254 & 46.432  & 10.511  \\ \hline 
Asia         & 0.193  & 3.838 & 2.251  & 20.862  & 3.206  \\ \hline       
\end{tabular}
}
\vspace{0.5cm}
\caption{The posterior mean values of the reassortment rate (per lineage per year), population size (year), height (year), total length (year) and reassortment node count across three regions.}
\label{tab:posterior_mean_values}
\end{table}


\subsection{Comparing the evolutionary histories of four viruses}

We examined the evolutionary histories of four viral groups: seasonal influenza A/H1N1, pandemic influenza A/H1N1, influenza A/H3N2, and influenza B by comparing the topologies of their inferred phylogenetic networks following the protocol of \citet{Muller2020} (excluding influenza A/H2N2 because it was not available in the repository). These datasets include two different types of influenza viruses, influenza A and influenza B. Influenza A viruses are classified into subtypes based on haemagglutinin (HA) and neuraminidase (NA) glycoproteins. We consider using the major subtypes H1N1 and H3N2 here. The H1N1 viruses are further divided into two distinct lineages originating from the 1918 and 2009 pandemics, where seasonal influenza A/H1N1 is descendant of the 1918 pandemic lineage and pandemic influenza A/H1N1 corresponds to the 2009 pandemic lineage. These two viruses differ greatly in their genomes (\citet{Smith2009}). For each viral group, all eight segments (PB2, PB1, PA, HA, NP, NA, M, NS) were analyzed. The seasonal influenza A/H1N1 includes 200 taxa sampled between 1995 and 2009, the pandemic influenza A/H1N1 includes 200 taxa sampled between 2009 and 2017, the influenza A/H3N2 includes 200 taxa sampled between 1988 and 2019, and the influenza B includes 200 taxa sampled between 1989 and 2019. Because sequence sampling times vary, the inferred networks are heterochronous, and thus, to compute pairwise distances among these networks, we applied the $h\mathbf{F}$-matrix encoding together with the heterochronous version of our metric, $d_2^{h}$, which accounts for branch lengths and heterochronous sampling times.
Pairwise distances were calculated using our weighted $d_2^{h}$ metric and visualized via MDS (Figure \ref{fig:coalre_data_4viruses}). The clear separation of the four viral groups in the MDS plot indicates that our metric effectively captures differences in their underlying evolutionary histories.


\begin{figure}[H]
\includegraphics[width=0.55\textwidth]{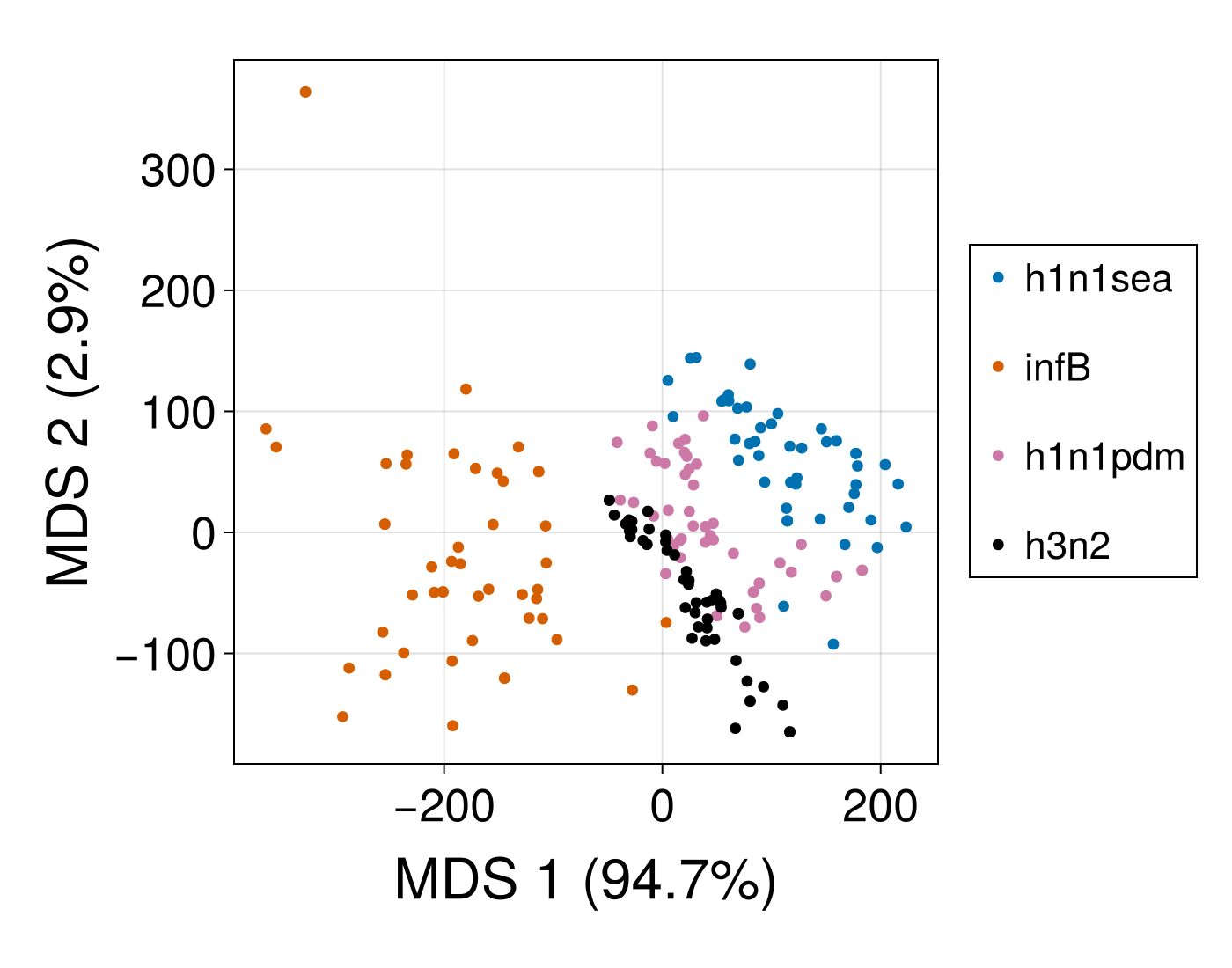}
\centering
\caption{Comparison of the evolutionary histories of seasonal influenza A/H1N1, pandemic influenza A/H1N1, influenza A/H3N2 and influenza B. Multiple dimensional scaling plot shows the four different network topologies group. This implies these four viruses have different evolutionary histories.}
\label{fig:coalre_data_4viruses}
\end{figure}

\section{Discussion}

Here, we developed a set of distance metrics for rooted, ranked, unlabeled phylogenetic networks that extend the $\mathbf F$-matrix approach from trees to account for hybridizations. We rigorously proved the bijection of our embedding, establishing a one-to-one correspondence between the space of phylogenetic networks and the associated $\mathbf F$-matrix space. 
Building on this foundation, we extended the metrics to compute pairwise distances between rooted, unlabeled, timed, binary, and isochronous phylogenetic networks by incorporating branch lengths, and defined weighted metrics $d_1^{w}$ and $d_2^{w}$ using the $L_1$ and $L_2$ norms, respectively. Because empirical analyses, such as those of influenza datasets, often involve taxa collected at different times, we further extended the framework to rooted, unlabeled, timed, binary, and heterochronous networks via a $h\mathbf{F}$-encoding method, defining $d_1^{h}$ and $d_2^{h}$ as the $L_1$ and $L_2$ distances of the corresponding encodings. 

To evaluate the performance of our metrics, we first simulated phylogenetic networks with varying topologies using a recently proposed beta-splitting model \citep{zhong2026beta}. Pairwise distances among these networks were computed using our $d_2$ metric and compared to alternative network distances, including the Hamming distance, hardwired cluster distance, and diffusion distance. Multidimensional scaling (MDS) was then applied to visualize the networks, with each point representing a single network. Our $d_2$ metric demonstrated superior ability to distinguish different topology distribution and achieved the fastest computational running time among the compared methods.  
We further applied our framework to empirical data, comparing the evolutionary histories of A/H1N1 influenza viruses across multiple geographic regions, as well as four distinct viral datasets, using the weighted $d_2^{w}$ and heterochronous $d_2^{h}$ metrics. In real influenza data, sampling at the same date often produces low divergence, resulting in inferred networks with few hybridizations or even purely tree-like topologies. To account for heterochronous sampling, we employed the $h \mathbf{F}$-encoding method with the $d_2^{h}$ metric to compute pairwise distances among phylogenetic networks of influenza A/H1N1 (seasonal and pandemic), influenza A/H3N2, and influenza B. The resulting MDS plot revealed clear clustering of the four viruses, indicating that their evolutionary histories are distinct and that our metrics can effectively capture biologically meaningful differences in reticulate evolutionary patterns.

Our metrics exhibit efficient computational performance and effectively discriminate among different phylogenetic network topologies. Nevertheless, several limitations warrant consideration. First, storing the matrix representations of networks can be memory-intensive, as the conversion of phylogenetic networks to triangular $\mathbf F$-matrices often results in large matrix sizes. Second, our $\mathbf F$-matrix space is formally defined as the set of matrices of equal dimensions. However, networks with differing numbers of speciation and hybridization events naturally map to matrices of varying sizes, placing them in distinct $\mathbf F$-matrix spaces. Although we implemented a practical strategy to align matrices of different hybridizations for distance computation, this remains a theoretical challenge when the number of samples varies. 

Looking forward, the matrix-based encoding provides a foundation for developing statistical summaries of the phylogenetic network space, including measures such as mean and variance, which could offer deeper insights into network topology distributions. 
Indeed, our proposed metrics have the potential to support statistical summaries of network samples. In particular, one could identify a matrix $F^*$ within the restricted space of matrices corresponding to ranked unlabeled networks that minimizes the sum of $d_2$ (or $d_1$) distances to a given sample. By the established bijection between these matrices and network topologies, $F^*$ would correspond to a ``mean'' or ``median'' network, providing a natural summary of the central tendency of the sample. This approach facilitates computation of summary statistics and formal comparison of sampling distributions, yet we do not explore it in the present work and leave its development and practical implementation as directions for future research.
Furthermore, as the first formal metric defined for unlabeled phylogenetic networks, our approach lays the groundwork for future methods and applications that require principled distances between such networks.




\section{Data and Code Availability}

Data used in comparing different regions experiment are from \href{https://gisaid.org/}{GISAID}. Full acknowledgment of all \href{https://gisaid.org/}{GISAID} data contributors is in gisaid\_supplemental\_table\_epi\_set\_260429ca.pdf, including accession numbers (EPI\_ISL IDs), originating and submitting laboratories, and associated authors for all sequences used in this study. Data are accessible via \href{https://gisaid.org/}{GISAID} EPI\_SET\_260429ca (DOI: \url{https://doi.org/10.55876/gis8.260429ca}).
All code and data needed to reproduce the results of this paper can be found in 
\url{https://github.com/JuliaPalacios/BirthHybrid_distance}.

\section{Acknowledgements}
J.A.P. acknowledges support from the NSF CAREER Award \#2143242 and NIH Award R35GM148338. 
This work was partly supported by the National Science Foundation (NSF CAREER DEB-2144367 to C.S.L.).

\newpage
\bibliographystyle{plainnat}
\bibliography{./networkspace_ref}

\end{document}